\documentclass[a4paper,UKenglish,cleveref, autoref,numberwithinsect,nameinlink]{lipics-v2021}
\hideLIPIcs
\nolinenumbers

\usepackage[table]{xcolor} 

\usepackage{pgfplots}
\usepackage{siunitx}
\usepackage{pgfplotstable}
\usepgfplotslibrary{groupplots}

\usepackage[utf8]{inputenc} 
\usepackage{amsmath}
\hypersetup{menucolor=orange!40!black,filecolor=magenta!40!black,urlcolor=blue!40!black,linkcolor=red!40!black,citecolor=green!40!black,colorlinks}
\usepackage{amssymb}
\usepackage{nicefrac}
\usepackage{todonotes}
\usepackage{dsfont}
\usepackage{aligned-overset}

\usepackage[sort&compress,numbers]{natbib}

\usepackage{booktabs}
\usepackage{multirow}

\usepackage[vlined,linesnumbered,ruled]{algorithm2e}
\SetEndCharOfAlgoLine{}
\SetCommentSty{textrm}
\SetKwProg{Fn}{Function}{}{}

\usepackage{etoolbox}

\newcommand{\appsymb}{$\star$}
\newcommand{\appref}[1]{{\hyperref[proof:#1]{\appsymb}}}

\theoremstyle{plain}
\theoremstyle{definition}

\newtheorem{rrule}{Reduction Rule}
\crefname{rrule}{Rule}{Rules}
\newtheorem{frule}{Forward Rule}[section]
\newtheorem{brule}{Backward Rule}[section]

\crefname{frule}{Forward Rule}{Forward Rules} %
\Crefname{frule}{FR}{FRs} %

\crefname{brule}{Backward Rule}{Backward Rules} %
\Crefname{brule}{BR}{BRs} %

\Crefname{theorem}{Theorem}{Theorems}
\crefname{theorem}{Thm.}{Thms.}
\Crefname{proposition}{Proposition}{Propositions}
\crefname{proposition}{Prop.}{Props.}
\Crefname{observation}{Observation}{Observations}
\crefname{observation}{Obs.}{Obs.}
\crefname{definition}{Def.}{Defs.}
\crefname{corollary}{Cor.}{Cors.}
\Crefname{corollary}{Corollary}{Corollaries}

\newcommand{\N}{\mathds{N}}

\newcommand{\RR}{\mathds{R}}
\newcommand{\Oh}{\ensuremath{\mathcal{O}}}
\newcommand{\infDef}{\textsf{Inflate-Deflate}}
\newcommand{\localInfDef}{\textsf{local Inflate-Deflate}}
\newcommand{\findMeth}{\textsf{Find}}
\newcommand{\findRed}{\textsf{Find and Reduce}}
\DeclareMathOperator{\cl}{cl}
\DeclareMathOperator{\vc}{vc}

\newcommand{\problemdef}[3]{
  \begin{center}
    \begin{minipage}{0.95\textwidth}
      \normalsize\textsc{#1} \smallskip \\
      \begin{tabularx}{\textwidth}{@{}l@{\hspace{3pt}}X}
        \normalsize\textbf{Input:}    & \normalsize#2 \\
        \normalsize\textbf{Question:} & \normalsize#3
      \end{tabularx}
    \end{minipage}
  \end{center}
}

\newcommand{\VC}[1]{{\normalfont\textsc{Vertex Cover}}}
\newcommand{\independentset}{\textsc{Independent Set}}

\title{There and Back Again: On Applying Data Reduction Rules by Undoing Others}

\author{Aleksander Figiel}{Technische Universität Berlin, Algorithmics and Computational Complexity, Germany}{a.figiel@tu-berlin.de}{}{Supported by DFG project ``MaMu'' (NI369/19).}

\author{Vincent Froese}{Technische Universität Berlin, Algorithmics and Computational Complexity, Germany}{vincent.froesen@tu-berlin.de}{}{}

\author{André Nichterlein}{Technische Universität Berlin, Algorithmics and Computational Complexity, Germany}{andre.nichterlein@tu-berlin.de}{https://orcid.org/0000-0001-7451-9401}{}

\author{Rolf Niedermeier}{Technische Universität Berlin, Algorithmics and Computational Complexity, Germany}{rolf.niedermeier@tu-berlin.de}{https://orcid.org/0000-0003-1703-1236}{}
\authorrunning{A.~Figiel, V.~Froese, A.~Nichterlein, and R.~Niedermeier}



\Copyright{Aleksander Figiel, Vincent Froese, André Nichterlein, and Rolf Niedermeier}

\begin{document}

\maketitle

\begin{abstract}
	Data reduction rules are an established method in the algorithmic toolbox for tackling computationally challenging problems. 
	A data reduction rule is a polynomial-time algorithm that, given a problem instance as input, outputs an equivalent, typically smaller instance of the same problem. 
	The application of data reduction rules during the preprocessing of problem instances allows in many cases to considerably shrink their size, or even solve them directly.
	Commonly, these data reduction rules are applied exhaustively and in some fixed order to obtain irreducible instances. 
	It was often observed that by changing the order of the rules, different irreducible instances can be obtained. 
	We propose to ``undo'' data reduction rules on irreducible instances, by which they become larger, and then subsequently apply data reduction rules again to shrink them. 
	We show that this somewhat counter-intuitive approach can lead to significantly smaller irreducible instances.
	The process of undoing data reduction rules is not limited to ``rolling back'' data reduction rules applied to the instance during preprocessing. 
	Instead, we formulate so-called backward rules, which essentially undo a data reduction rule, but without using any information about which data reduction rules were applied to it previously. 
	In particular, based on the example of \VC{} we propose two methods applying backward rules to shrink the instances further. 
	In our experiments we show that this way smaller irreducible instances consisting of real-world graphs from the SNAP and DIMACS datasets can be computed.
\end{abstract}

\newpage

\section{Introduction}\label{sec:intro}

Kernelization by means of applying data reduction rules is a powerful (and often essential) tool for tackling computationally difficult (e.\,g.\ NP-hard) problems in theory and in practice~\cite{Fom19,ALMNSS20}. 
A data reduction rule is a polynomial-time algorithm that, given a problem instance as input, outputs an ``equivalent'' and often smaller instance of the same problem.
One may think of data reduction as identifying and removing ``easy'' parts of the problem, leaving behind a smaller instance containing only the more difficult parts.
This instance can be significantly smaller than the original instance~\cite{ACF04,AI16,KKN21,GLS21,Wei98} which makes other methods like branch\&bound algorithms a viable option for solving it.
In this work, we apply existing data reduction rules ``backwards'', that is, instead of smaller instances we produce (slightly) larger instances. 
The hope herein is, that this alteration of the instance allows subsequently applied data reduction rules to further shrink the instance, thus, producing even smaller instances than with ``standard'' application of data reduction rules.

We consider the NP-hard \VC{}---the primary ``lab animal'' in parameterized complexity theory~\cite{FJK18}---to illustrate our approach and to exemplify its strengths.
\problemdef{\VC{} \cite{GJ79}}
{An undirected graph~$G = (V,E)$ and~$k\in \N$.}
{Is there a set~$S \subseteq V$, $|S| \le k$, covering all edges, i.\,e., $\forall e \in E\colon e \cap S \neq \emptyset$?}
\VC{} is a classic problem of computational complexity theory and one of Karp's~21 NP-complete problems~\cite{Kar72}.
We remark that for presentation purposes we use the decision version of \VC{}. 
All our results transfer to the optimization version (which our implementation is build for).

To explain our approach, assume that all we have is the following data reduction rule:
\begin{rrule}[Triangle Rule {\cite{FJK18}}]\label{rr:triangle}
	Let~$(G=(V,E),k)$ be an instance of \VC{} and~$v \in V$ a vertex with exactly two neighbors~$u$ and~$w$. 
	If the edge~$\{u,w\}$ exists, then delete~$v$, $u$, and~$w$ from the graph (and their incident edges), and decrease~$k$ by two.
\end{rrule}
\begin{figure}[t]
	\centering
	\begin{subfigure}[c]{.55\textwidth}
		\centering
		\includegraphics{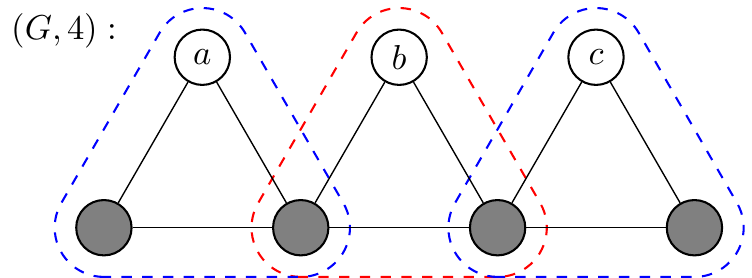}
		\caption{An instance~$(G,4)$ for \VC{}. A vertex cover of size four is indicated by the gray vertices.}
		\label{fig:bad_triangles_a}
	\end{subfigure}~~~~
	\begin{subfigure}[c]{.35\textwidth}
		\centering
		\includegraphics{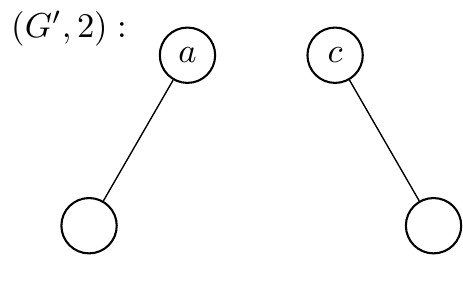}
		\caption{The outcome $(G',2)$ of applying the Triangle Rule to~$b$.}
		\label{fig:bad_triangles_b}
	\end{subfigure}
	\caption{
		An example of how the order of reduction rules can affect the final instance. 
		A \VC{} instance with $k=4$ is depicted left. 
		By applying the Triangle Rule to~$b$ (red triangle), we obtain a graph with two edges (right). 
		However, applying the rule to~$a$ and then~$c$ (therefore removing the two blue triangles) yields the smaller graph consisting only of the vertex~$b$.}
	\label{fig:bad_triangles}
\end{figure}
As illustrated in \cref{fig:bad_triangles}, there are two options to apply \cref{rr:triangle} for the instance~$(G,4)$.
Picking the ``bad'' option, that is, applying it to~$b$ yields the instance~$(G',2)$.
Note that (the correctness of) \cref{rr:triangle} implies that~$(G,4)$ and~$(G',2)$ are ``equivalent'', that is, either both of them are yes-instances or none of them are.
Hence, if we have the instance~$(G',2)$ on the right side (either through the ``bad'' application of \cref{rr:triangle} or directly as input), then we can apply \cref{rr:triangle} \emph{``backwards''} and obtain the equivalent but larger instance~$(G,4)$ on the left side.
Then, by applying \cref{rr:triangle} to~$a$ and~$c$, we can arrive at the edge-less graph~$(\{b\},\emptyset)$, thus ``solving'' the triangle-free instance~$(G',2)$ by only using the Triangle Rule.

More formally, the setting can be described as follows:
A data reduction rule for a problem~$L$ is a polynomial-time algorithm which \emph{reduces} an instance~$x$ to an equivalent instance~$x'$, that is,~$x\in L$ if and only if~$x'\in L$.
A set of data reduction rules thus implicitly partitions the space of all instances into classes of equivalent instances (two instances are in the same equivalence class if one of them can be obtained from the other by applying a subset of the data reduction rules).
The more data reduction rules we have, the fewer and larger equivalence classes we have.
Now, the overall goal of data reduction is to find the \emph{smallest} instance in the same equivalence class.
We demonstrate two approaches tailored towards (but not limited to) graph problems to tackle this task.

Let us remark that while there are some analogies to the branch\&bound paradigm (searching for a solution in a huge search space), there are also notable differences:
A branching rule creates several instances of which at least one is guaranteed to be equivalent to the original one.
The problem is that, a priori, it is not known which of these instances is the equivalent one. 
Hence, one has to ``solve'' all instances before learning the solution.
In contrast, our setting allows stopping at \emph{any} time as the currently handled instance is \emph{guaranteed} to be equivalent to the starting instance.
This allows for considerable flexibility with respect to possible combinations with other approaches like heuristics, approximation or exact algorithms.

\subparagraph*{Related Work.}
\citet{FJK18} are closest to our work.
They propose a method for automated discovery of data reduction rules looking at rules that replace a small subgraph by another one.
They noticed that if a so called \emph{profile} (which is a vector of integers) of the replaced subgraph and of the one taking its place only differ by a constant in each entry, then this replacement is a data reduction rule.
To then find data reduction rules, one can enumerate all graphs up to a certain size and compute their profile vectors. 
The downside of this approach is that in order to apply the automatically found rules one has solve a (computationally challenging) subgraph isomorphism problem or manually design new algorithms for each new rule. 
\VC{} is extensively studied from the the viewpoint of data reduction and kernelization; see \citet{FJK18} for an overview. 
\citet{AI16} and \citet{HLS20} provide exact solvers that include an extensive list of data reduction rules. 
The solver of \citet{HLS20} won the exact track for \VC{} at the 4th PACE implementation challenge~\cite{DFH19}.
A list of data reduction rules for \VC{} is provided in \cref{chap:vc_rrs}. 

\citet{AHL03} experimentally investigated by how much the so-called Struction data reduction rule for \independentset{} can shrink small random graphs.
The Struction data reduction rule can always be applied to any graph and decreases the stability number\footnote{An independent set is a set of pairwise nonadjacent vertices. The stability number or the independence number of a graph~$G$ is the size of a maximum independent set of~$G$.} of a graph by one, but may increase the number of vertices quadratically each time it is applied.
\citet{GLS21} proposed a modification of the Struction rule for the \textsc{Maximum Weighted Independent Set} problem.
They first restricted themselves to only applying data reduction rules if they do not increase the number vertices in the graph, which they call the reduction phase.
They then compared this method to an approach which allows their modified Struction rule to also increase the number of vertices in the graph by a small fraction, which they call the blow-up phase.
The experiments showed, that repetitions of the reduction and blow-up phase can significantly shrink the number of vertices compared to just the reduction phase.

\citet{Ehrig14} defined the notion of confluence from rewriting systems theory for kernelization algorithms.
Intuitively, confluence in kernelization means that the result of applying a set of data reduction rules exhaustively to the input always results in the same instance, up to isomorphism, regardless of the order in which the rules were applied.
It turns out that for our approach to work we require non-confluent data reduction rules.

\subparagraph*{Our Results.}
In \cref{sec:methods}, we provide two concrete methods to apply existing data reductions rules ``backwards'' and ``forwards'' in order to shrink the input as much as possible. 
We implemented these methods and applied them on a wide range of data reduction rules for \VC{}. %
Our experimental evaluations are provided in \cref{sec:exp} where we use our implementation on instances where the known data reductions rules are not applicable. 
Our implementation can also be used to preprocess a given graph~$G$ and it returns the smallest found kernel~$K$ after a user specified amount of time. 
Moreover, the implementation can translate a provided solution~$S_K$ for the kernel into a solution~$S_G$ for the initial instance such that~$|S_G| \le |S_K| + d$ where~$d := \tau(G) - \tau(K)$ is the difference between the sizes of minimum vertex covers of~$G$ and~$K$.
Thus, if a minimum vertex cover for~$K$ is provided it will be translated into a minimum vertex cover for~$G$.

\section{Preliminaries}

We use standard notation from graph theory and data reduction.
In this work, we only consider simple undirected graphs~$G$ with vertex set~$V(G)$ and edge set~$E(G) \subseteq \{ \{v, w\} \mid v, w \in V(G), v \neq w \}$.
We denote by~$n$ and~$m$ the number of vertices and edges, respectively.
For a vertex~$v \in V(G)$ the open (closed) neighborhood is denoted with~$N_G[v]$ ($N_G(v)$).
For a vertex subset~$S \subseteq V(G)$ we set~$N_G[S] := \bigcup_{v\in S} N_G[v]$.
When in context it is clear which graph is being referred to, the subscript~$G$ will be omitted in the subscripts.

\subparagraph*{Data Reduction Rules.}
We use notions from kernelization in parameterized algorithmics~\cite{Fom19}. 
However, we simplify the notation to unparameterized problems.
A data reduction rule for a problem~$L\subseteq \Sigma^*$ is a polynomial-time algorithm, which \emph{reduces} an instance~$x$ to an equivalent instance~$x'$.
We call an instance~$x$ \emph{irreducible} with respect to a data reduction rule, if the data reduction rule does not change the instance~$x$ any further (that is,~$x'=x$).
The property that the data reduction rule returns an equivalent instance is called \emph{safeness}.
We call an instance obtained from applying data reduction rules \emph{kernel}.

Often, data reduction rules can be considered nondeterministic, because a data reduction rule could change the input instance in a variety of ways (e.\,g., see the example in \cref{sec:intro}).
To highlight this effect and to avoid confusion, we introduce the term \emph{forward rule}.
A forward rule is a subset $R_\mathcal{A}\subseteq \Sigma^*\times \Sigma^*$ associated with a nondeterministic polynomial-time algorithm~$\mathcal{A}$, where~$(x, y) \in R_\mathcal{A}$ if and only if~$y$ is one of the possible outputs of~$\mathcal{A}$ on input~$x$.
Intuitively, a forward rule captures all possible instances that can be derived from the input instance by applying a data reduction rule a single time.
To define what it means to ``undo'' a reduction rule, we introduce the term \emph{backward rule}.
A backward rule is simply the converse relation~$R^{-1}_{\mathcal{A}} := \{ (y,x) \mid (x,y) \in R_\mathcal{A}\}$ of some forward rule~$R_\mathcal{A}$.

\subparagraph*{Confluence.}
A set of data reduction rules is said to be terminating, if for all instances~$\mathcal{I}$, the data reduction rules in the set cannot be applied to the instance~$\mathcal{I}$ infinitely many times.
A set of terminating data reduction rules is said to be confluent if any exhaustive way of applying the rules yields a unique irreducible instance, up to isomorphism~\cite{Ehrig14}.
It is not hard to see that given a set of confluent data reduction rules, undoing any of them is of no use, because subsequently applying the data reduction rules will always result in the same instance.
\begin{lemma}%
	\label{lem:confluence_bad}
	Let~$\mathcal{R}$ be a confluent set of forward rules, $\mathcal{I}$ a problem instance, and~$\mathcal{I}^\mathcal{R}$ the unique instance obtained by applying the rules in~$\mathcal{R}$ to~$\mathcal{I}$.

	Further let~$\overline{\mathcal{R}} = \{R^{-1} \mid R \in \mathcal{R}\}$ be the set of backward rules corresponding to~$\mathcal{R}$.
	Let~$\mathcal{I}^-$ be an instance that was derived by applying some rules from~$\mathcal{R} \cup \overline{\mathcal{R}}$.
	Then applying the rules in~$\mathcal{R}$ exhaustively in any order to~$\mathcal{I}^-$ will yield an instance isomorphic to~$\mathcal{I}^\mathcal{R}$.
\end{lemma}
{
\begin{proof}
	We prove the lemma by induction over the number of times backward rules were applied to obtain~$\mathcal{I}^-$.

	Base case: if no backward rules were applied to obtain~$\mathcal{I}^-$, then by exhaustively applying~$\mathcal{R}$ we will obtain an instance isomorphic to~$\mathcal{I}^\mathcal{R}$.

	Inductive step: Assume only~$n$ backward rules were applied to obtain~$\mathcal{I}^-$.
	Consider the sequence~$\mathcal{I} = \mathcal{I}_1, \mathcal{I}_2 \dots, \mathcal{I}_{\ell-1}, \mathcal{I}_\ell = \mathcal{I}^-$ of instances which were on the way from~$\mathcal{I}$ to~$\mathcal{I}^-$ during the application of the rules in~$\mathcal{R} \cup \overline{\mathcal{R}}$.
	Further let~$\mathcal{I}_i$ be the instance in the sequence after applying the~$n$'th backward rule~$R^{-1}$.
	After applying~$R$ to~$\mathcal{I}_i$ the instance~$\mathcal{I}_{i-1}$ can be obtained, which was obtained by using only~$n-1$ backward rules.
	Because after~$R^{-1}$ only rules in~$\mathcal{R}$ were applied, which are confluent, we may instead assume that the first rule which was applied after~$R^{-1}$ was~$R$.
	Consequently, by the induction hypothesis an instance isomorphic to~$\mathcal{I}^\mathcal{R}$ will be obtained by exhaustively applying the rules in~$\mathcal{R}$.
\end{proof}
}
\section{Two Methods for Achieving Smaller Kernels} \label{sec:methods}

We apply data reduction rules ``back and forth'' to obtain an equivalent instance as small as possible.
This gives rise to a huge search space for which exhaustive search is prohibitively expensive.
Thus, some more sophisticated search procedures are needed. 
In this section, we propose two approaches which we call the \findMeth{} and the \infDef{} method.
We implemented and tested both approaches; the experimental results are presented in \cref{sec:exp}.

The \findMeth{} method (\cref{ssec:find}) shrinks the naive search tree with heuristic pruning rules in order to identify sequences of forward and backward rules, which when applied to the input instance, produce a smaller equivalent instance.
We employ this method primarily to find such sequences which are \emph{short}, so that those sequences may actually be used to formulate \emph{new} data reduction rules.
It naturally has a \emph{local} flavor in the sense that changes of one iteration are bound to a (small) part of the input graph.

The \infDef{} method (\cref{ssec:infDef}) is much less structured.
It randomly applies backward rules until the instance size increased by a fixed percentage. 
Afterwards, all forward rules are applied exhaustively. 
If the resulting instance is smaller, then the process is repeated; otherwise, all changes are reverted.

\subsection{\findMeth{} Method}\label{ssec:find}

For finding sequences of forward and backward rules which when applied to the input instance produce a smaller instance, we propose a structured search approach based on recursion.
Let~$\mathcal{I}$ be the input instance and let~$\mathcal{F}_{\mathcal{I}}$ be the set of all instances reachable via one forward or backward rule, i.\,e., $\mathcal{F}_{\mathcal{I}} = \{\mathcal{I}' \mid (\mathcal{I}, \mathcal{I}') \in R \text{ for any forward or } \allowbreak \text{backward rule } R\}$.
If the input instance is irreducible with respect to the set of forward rules, then only backward rules will be applicable.
We branch into~$|\mathcal{F}_{\mathcal{I}}|$ cases where, for each~$\mathcal{I}' \in \mathcal{F}_{\mathcal{I}}$, we try to recursively find forward and backward rules applicable to~$\mathcal{I}'$ and branch on each of them. 
This is repeated until some maximum recursion depth is reached or the sequence of forward and backward rule applications results in a smaller instance.
In the latter case, the sequence of applied rules can be thought of as a new data reduction rule.

Note that the search space is immense:
For example, consider the backward rule corresponding to \cref{rr:triangle}, which inserts three vertices~$u$, $v$, and~$w$, makes them pairwise adjacent, and inserts an arbitrary set of edges between~$u$ and~$v$ and the original vertices.
Thus, for each original vertex, there are four options (make it adjacent to~$u$, to~$v$, to~$u$ and~$v$, or neither).
This results in~$4^n$ options for applying just this one single backward rule.
Hence, it is clear that we have to introduce suitable methods to cut off large parts of the resulting search tree.
To this end, we heavily rely on the observation that many reduction rules have a ``local flavor''.

\subparagraph*{Region of Interest.}
To avoid a very large search space we only consider applying forward and backward rules ``locally''.
For this, we define a ``region of interest'', which for graph problems is a set~$X \subseteq V$.
Any forward or backward rule must only be applied within the region of interest.
We initially start with a very modest region of interest, namely~$X = \{v\}$ for all~$v \in V$. 
Thus, we work with~$n$ regions of interest per graph, each one considered separately.

Forward and backward rules are allowed to leave the region of interest only if at least one vertex that is ``relevant'' for the rule is within the region of interest.
Moreover, the region of interest is allowed to ``grow'' as rules are applied.
This is because applications of rules might cause further rules to become applicable.
For example, rules might become applicable to the neighbors of vertices modified by the previously applied rules.

Specifically, let~$M$ be the set of ``modified'' vertices, which are new vertices or vertices which gained or lost an edge as a result of applying a forward or backward rule, and let~$D$ be the set of vertices it deleted.
In the case of \VC{}, we suggest to expand the region of interest to~$(X \cup N[M]) \setminus D$ after each rule application.
The majority of forward rules for \VC{} (see \cref{chap:vc_rrs}) are ``neighborhood based''.
Of course, the region of interest could be expanded even further, e.\,g., by extending it by~$N^2[M]$ instead, but this will of course increase the search space.
With a larger region of interest we might find more reduction rules, but at the cost of higher running time.
Additionally, the found reductions may be more complex, and modify a large subgraph and are therefore difficult to analyze or implement.

We will call the above method, which recursively applies forward and backward rules one by one restricted to only the region of interest, the \findMeth{} method.
Note that we ``accept'' a sequence of forward and backward rules if it decreases the number of vertices or the parameter~$k$, without increasing either.
We have done so to find only ``nice'' data reduction rules where there is no trade-off between decreasing the parameter~$k$ or the number of vertices.
However, different conditions to ``accept'' a sequence of forward and backward rules are also possible.
For example, by only requiring that the number of vertices or edges is decreased.

\subparagraph*{Graph modification.}

Our implementation of the \findMeth{} method outputs sequences of rules which are able to shrink the graph.
However, just knowing which rules and in what order they were applied may not be very helpful in understanding the changes made by the rules.
Specifically, it does not show how the rules were applied.

For this reason, we introduce the notion of a \emph{graph modification} based on the ideas by \citet{FJK18}.
A graph modification encodes how a single or multiple data reduction rules have changed a graph.
We say the \emph{boundary} of a subgraph~$H$ of a graph~$G$ is the set~$B$ of all vertices in~$V(H)$ whose neighborhood in~$G$ contains vertices in~$V(G) \setminus V(H)$.
\begin{definition}
A graph modification is a 4-tuple of graphs~$(G, H, H', G')$ with the following properties
\begin{itemize}
	\item~$B = V(H') \cap V(G)$,
	\item~$H$ is a subgraph of~$G$ with boundary~$B$, and
	\item~$G'$ is derived from~$G$ by deleting all vertices in~$V(H) \setminus B$ and all edges among~$B$ from~$G$ and then adding the vertices in~$V(H') \setminus B$ and adding all edges from~$H'$.
\end{itemize}
\begin{figure}[t]
	\centering
	\begin{subfigure}[c]{.35\textwidth}
		\centering
		\includegraphics{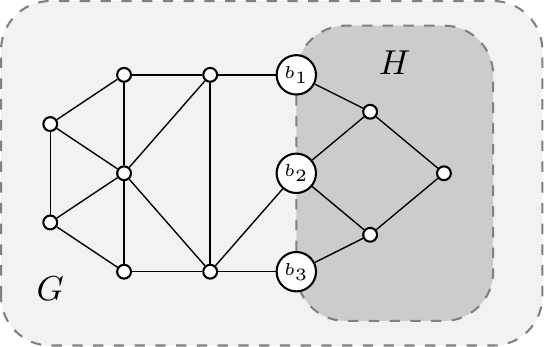}
	\end{subfigure}~~~
	\hspace*{0.7em}
	\begin{subfigure}[c]{.15\textwidth}
		\centering
		\includegraphics{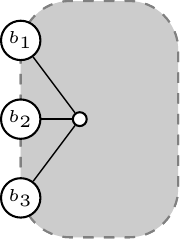}
	\end{subfigure}~
	\begin{subfigure}[c]{.35\textwidth}
		\centering
		\includegraphics{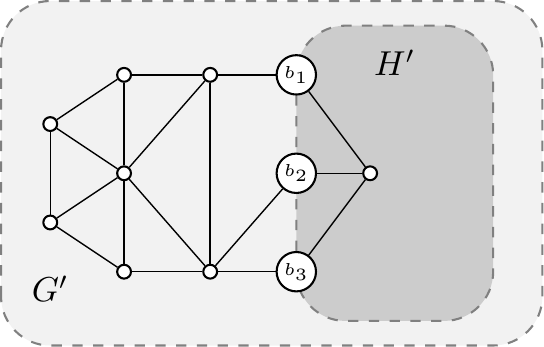}
	\end{subfigure}
\caption{A graph modification example corresponding to the application of the Degree-2 Folding Rule (\cref{rr:deg2}) that applies to vertices with exactly two neighbors that are non-adjacent. The set~$B = \{b_1, b_2, b_3\}$ is the boundary of the graph modification. Notice that only vertices in~$B$ (the boundary) are adjacent to both vertices in~$V(H)$ and~$V(G) \setminus V(H)$}
\label{fig:graph_modification}
\end{figure}
\end{definition}
See \cref{fig:graph_modification} for an example of a graph modification.
The \findMeth{} method can be extended such that in addition to printing rule sequences it also outputs the graph modifications corresponding to each application of a rule from the found sequences.
This can be achieved by keeping track of newly created or deleted vertices and edges by the applied rules.

\subparagraph*{Isomorphism.}
It may happen that two or more rule sequences change an instance in the same way, that is, they produce isomorphic graphs from the same input instance.
{
One reason is that often regions of interest~$X$ can substantially overlap, and the same reduction rule sequence could be found in overlapping regions of interest.
Another reason is that one rule A could generalize a different rule B, so if a rule sequence containing B is found, then the same rule sequence with A instead of B will also be found.

It can therefore happen that some sequences are printed too often.
This could then potentially disturb the statistics about the frequency distribution of different sequences.
One approach to solve this small problem is to not apply a rule if a different rule which it generalizes is applicable.

A more flexible approach is to not accept a reduction rule sequence which produced an instance~$(G',k')$ if a different sequence produced an instance~$(G'', k')$ where~$G' \cong G''$.

The isomorphism test for large graphs can be a challenging task.
The graph isomorphism problem is known to be in NP, however, it is also not known to be NP-hard or in P~\cite{Mat79}.
In practice the (1-dimensional) Weisfeiler-Lehman Algorithm~\cite{WL68} (also known as Color-Refinement) is able to quickly distinguish almost all graphs~\cite{BES80}.
However, the test can only distinguish some non-isomorphic graphs, therefore it cannot be used to prove isomorphism.

We propose a modified isomorphism test in which we additionally exploit the observation that after applying only a few rules, the major parts of a graph are unchanged.
For this we require the following definitions:

\begin{definition}
	Two graph modifications~$(G, H_1, H_1', G_1)$ and~$(G, H_2, H_2', G_2)$ are called \emph{locally isomorphic} if
	there exists a bijective function~$f: V(G_1) \mapsto V(G_2)$ with the following properties
	\begin{itemize}
		\item~$\{a, b\} \in E(G_1)$ if and only if~$\{f(a), f(b)\} \in E(G_2)$, and
		\item for all~$v \in V(G) \setminus ((V(H_1) \cup V(H_2))$ the vertex~$v$ is mapped to itself, i.\,e., $f(v) = v$.
	\end{itemize}
\end{definition}

Intuitively, two graph changes are locally isomorphic if they change the graph in the same way.
Note that local isomorphism corresponds to standard graph isomorphism testing if~$V(H_1) \cup V(H_2) = V(G)$.

In this modified isomorphism definition we enforce that all unmodified vertices of the two graphs are mapped to each other.
This can be implemented by adapting the Individualization-Refinement and Color-Refinement algorithms which are used for standard isomorphism testing~\cite{CG70,MP14}.
The Individualization-Refinement algorithm computes a so-called \emph{canonical labeling}~$\cl(G)$ of the vertices of a graph such that two graphs~$G,H$ are isomorphic if and only if~$\cl(G) = \cl(H)$.
Therefore, after computing the canonical labeling the task of determining isomorphism becomes a simple problem of testing for equality.
The only modification that needs to be made is that we fix a labeling for all vertices not in~$V(H_1) \cup V(H_2)$, and do not change these labels within Individualization-Refinement and Color-Refinement, which is used as a subroutine for the former.
Furthermore, it can be shown that the entire graph~$G$ does not need to be considered, as the two algorithms update the labels only based on the labels of direct neighbors. 
This modified algorithm often significantly outperforms the full isomorphism test, as already most vertices have a unique label already assigned to them.

There is still however the case of ``redundancy'' within reduction rule sequences.
For example a vertex could be split using \cref{rr:undeg2} and then immediately merged again using \cref{rr:deg2}. 
Such cases can also be discarded by testing for local isomorphism to one of the graph modifications corresponding to only applying a reduction rule sequence which is a prefix of the current reduction rule sequence.

Local isomorphism testing can be further used to find ``minimal'' forward and backward rule sequences.
If two graph modifications~$M_1$ and~$M_2$ corresponding to two different rule sequences are locally isomorphic and one sequence is longer than the other then that sequence is not minimal.
}

\subparagraph*{The \findRed{} Method.}\label{sect:FAR}

The \findRed{} Method is just a small variation of \findMeth{}.
Instead of only searching for sequences of forward and backward rules which shrink the instance, upon finding such a sequence it is also immediately applied to the instance.
The search for more sequences continues with the smaller instance.
This method serves the dual purpose of both finding sequences of rules which shrink the instance (and therefore also finding reduction rules), but also that of producing a smaller irreducible instance.
Another potential advantage of this method is that it finds only the rules which were used to produce the smaller instance, and may therefore be more practical than the ones found by \findMeth{}.
Recall, that \findMeth{} only searches for reduction rules applicable directly to the input instance.
Perhaps by applying a single such sequence, different sequences are needed to shrink the remaining instance further.

\subsection{\infDef}\label{ssec:infDef}

Both the \findMeth{} and the \findRed{} method 
only change a small part of the instance.
It may, however, be necessary to change large parts of an instance before it can be shrunk to a size smaller than it was originally.
For this reason, we propose the \infDef{} inspired by the Cyclic Blow-Up Algorithm by \citet{GLS21}.

Essentially, \infDef{} iteratively runs two phases:
First, in the inflation phase, randomly applies a set of backward rules to the instance until it becomes some fixed percentage~$\alpha$ larger than it was initially. 
Then, in the deflation phase, exhaustively apply a set of forward rules and repeat with the inflation phase again.
Our implementation has the termination condition~$|V| = 0$ which may never be met.
Thus a timeout or limit on the number of iterations has to be specified.
The inflation factor~$\alpha$ can be freely set to any value greater zero.
We investigate the effect of different inflation factors in \cref{sec:exp} for values of~$\alpha$ between 10\% and 50\%.

In our deflate procedure, we apply the set of forward rules exhaustively in a particular way:
An applicable forward rule is randomly chosen, and then it is randomly applied to the instance, but only once.
Afterwards another applicable rule is chosen randomly, and this is repeated until the instance becomes irreducible with respect to all forward rules.
This randomized exhaustive application of the forward rules ensures that the rules are not applied in a predefined way, and different ``interactions'' of the rules are tested.
For example, consider the case where in the inflate phase the Backward Degree-2 Folding Rule (\Cref{rr:deg2}) was applied, then likely one would not want to immediately exhaustively apply the Degree-2 Folding Rule which could in effect directly cancel the changes made by that backward rule.
Furthermore, in this way, each of the forward rules has a chance to be applied.
This avoids any potential problems due to an inconvenient fixed rule order.

Because large sections of an instance are modified at once, no short sequences of forward and backward rules can be extracted from this method.
As a result, it is unlikely that new reduction rules could be learned this way.
However, the method produces smaller irreducible instances as we will see in \cref{sec:exp}.

\subparagraph*{Local Inflate-Deflate.}
Within \infDef{} it may happen that if we inflate the instance and then deflate it again, often the resulting instance is larger.
In such cases the number of ``negative'' changes to the instance outweigh the number of ``positive'' changes.
To increase the success probability one may try to lower the inflation factor, however then it can also happen that positive changes are less likely.

An alternate way to try to increase the success probability, is to apply backward rules within a randomly chosen subgraph rather than the whole graph.
For example, this subgraph could be the set of all vertices with some maximum distance to a randomly chosen vertex.
Backward rules are then applied until the subgraph becomes larger by a factor of~$\alpha$, instead of the whole graph.

\section{Experimental Evaluation} \label{sec:exp}

In this section, we describe the experiments which we performed based on an implementation of the methods described in \cref{sec:methods}.

\subsection{Setup}

\subparagraph*{Computing Environment.}
All our experiments were run on a machine running Ubuntu 18.04 LTS with the Linux 4.15 kernel.
The machine is equipped with an Intel\textsuperscript{\textregistered{}} Xeon\textsuperscript{\textregistered{}} W-2125 CPU, with 4 cores and 8 threads\footnote{All our implementations are single-threaded.} clocked at 4.0 GHz and 256GB of RAM.

\subparagraph*{Datasets.}
For our experiments we used three different datasets: DIMACS, SNAP and PACE; the lists of graphs are given in \cref{chap:tables}.
The DIMACS and SNAP datasets are commonly used for graph-based problems, including \VC{} \cite{AI16,KKN21,GLS21}.
We have used the instances from the 10th DIMACS Challenge~\cite{DIMACS10}, specifically from the Clustering, Kronecker, Co-author and Citation, Street Networks, and Walshaw subdatasets.
In total these are 82 DIMACS instances.
From the SNAP Dataset Collection~\cite{SNAP} we have used the graphs from from the Social, Ground-Truth Communities, Communication, Collaboration, Web, Product Co-purchasing, Peer-to-peer, Road, Autonomous systems, Signed and Location subdatasets.
In total we obtained 52 SNAP instances.
Additionally, we used a dataset which was used specifically for benchmarking \VC{} solvers in the 2019 PACE Challenge~\cite{PACE}.
We used the set of 100 private instances~\cite{PACEDATA}, which were used for scoring submitted solvers.

\begin{table}[t]
	\caption{
		A glossary of the forward and backward rules which were used in our implementation. 
		Note that we apply the Struction Rule only if it does not increase~$k$. 
		In the columns named alias we provide shortened names for the rules.
	}
	\centering
	\footnotesize
	\renewcommand{\tabcolsep}{5pt}
	\begin{tabular}{ll l  ll l}
		\toprule
		\multicolumn{3}{c}{Forward rules} & \multicolumn{3}{c}{Backward rules} \\
		Alias & Full name & Ref. & Alias & Full name & Ref. \\
		\midrule
		Deg0 		& Degree-0\hspace{-0.5em} 				& \Cref{rr:deg0}	& Undeg2 & Backward Degree-2 Folding\hspace{-1em} & \Cref{rr:undeg2}\\
		Deg1 		& Degree-1\hspace{-0.5em} 				& \Cref{rr:deg1}	& Undeg3 & \multirow{2}{4cm}{Backward Degree-3 Independent Set} & \Cref{rr:undeg3}\\
		Deg2		& Degree-2 Folding\hspace{-0.5em}			& \Cref{rr:deg2}	&&\\ 
		Deg3		& Degree-3 Independent Set\hspace{-0.5em}		& \Cref{rr:deg3}	& Uncn & \multirow{2}{4cm}{Backward 2-Clique Neigh-\\borhood (special case)}& \Cref{rr:uncn}\\ 
		Dom		& Domination\hspace{-0.5em}				& \Cref{rr:domination}	&&\\ 
		Unconf\hspace{-0.5em} & Unconfined-$\kappa$ ($\kappa = 4$)\hspace{-0.5em}	& \Cref{rr:unconfinedpp}& Undom & Backward Domination & \Cref{rr:undom} \\ 
		Desk		& Desk\hspace{-0.5em}					& \Cref{rr:desk}	& Ununconf\hspace{-2.5em} & Backward Unconfined & \Cref{rr:ununconf}\\ 
		CN 		& 2-Clique Neighborhood\hspace{-0.5em} 		& \Cref{rr:cn} 		& OE\_Ins & Optional Edge Insertion & \Cref{rr:oe_insert}\\
		OE\_Del\hspace{-0.5em } & Optional Edge Deletion\hspace{-0.5em}		& \Cref{rr:oe_delete}	&&&\\ 
		Struct\hspace{-0.5em} & Struction ($k' \leq k$)\hspace{-0.5em}		& \Cref{rr:struction}	&&&\\ 
		Magnet\hspace{-0.5em} & Magnet\hspace{-0.5em}				& \Cref{rr:magnet}	&&&\\ 
		LP		& LP\hspace{-0.5em}					& \Cref{rr:lp}		&&&\\ 
		\bottomrule
	\end{tabular}
	\label{tab:forward_glossary}
\end{table}

\subparagraph*{Preprocessing and Filtering.}
We apply some preprocessing to our datasets.
We obtain simple, undirected graphs by ignoring any potential edge direction or weight information from the instances and by deleting self-loops.
To these graphs we apply the forward rules, see \cref{tab:forward_glossary,chap:vc_rrs} for an overview: 
Deg1, Deg2, Deg3, Unconf, Cn, LP, Struct, Magnet and Oe\_delete exhaustively in the given order.
We note that the kernels obtained this way always had fewer vertices than the kernels obtained with the data reduction suite used by \citet{AI16} and also \citet{HLS20}.

We filter out graphs that became empty as a result of applying these rules.
These were 41 DIMACS, 33 SNAP and 12 PACE instances.
Furthermore, we discard graphs which after applying these rules still had more than 50,000 vertices.
These were 12 DIMACS and 3 SNAP instances.
Because the PACE instances may also contain instances from the other two datasets, we have tested the graphs for isomorphism.
We have found one PACE instance to be isomorphic to a SNAP instance (p2p-Gnutella09), which was already excluded, because it was shrunk to an empty graph.

In total, we are left with 31 DIMACS, 16 SNAP and 88 PACE graphs with at most 50,000 vertices---all of these graphs are irreducible with respect to the forward rules.
When referring to the datasets DIMACS, SNAP and PACE we will be referring to these kernelized and filtered instances.
See \cref{chap:tables} for tables with basic properties of these graphs.

\subparagraph*{Implementation.} 
The major parts of our implementation are written in C++11 and compiled using version 7.5 of g++ using the -O2 optimization flag.
Smaller parts, such as scripts for visualization or automation were written in Python 3.6 or Bash.
We provide the source code for our implementation at 
\url{https://git.tu-berlin.de/afigiel/undo-vc-drr}. 
This implementation contains our \findMeth{} and \infDef{} method, together with the two small variations \findRed, and \localInfDef. 
\findMeth{} uses the local isomorphism test which we describe in \cref{ssec:find}, and output a description of the graph modification in addition to the found sequences.

We also provide a \VC{} solver implementation with all our forward rules implemented, and some new data reduction rules which are explained later in this section.
The solver is based on the branch-and-reduce paradigm and is very similar to the \VC{} solver by \citet{AI16}.
Moreover, we provide a lifting algorithm that can transform solutions for the kernelized instances into solution for the original instances.
We also provide a Python script that is used to visualize the graph modifications of the sequences of forward and backward rules that are output by our methods.
A sample visualization can be found in \cref{chap:tables}.
However, our focus in this section is on the \findMeth{} and \infDef{} method.

\subparagraph*{Methodology.}
We have implemented our \findMeth{} and \infDef{} methods together with their two variations: \findRed{} and \localInfDef.
Almost all forward and backward rules presented in \cref{chap:vc_rrs} are used by these methods, see \cref{tab:forward_glossary} for an overview.

All rules increase or decrease~$k$, but do not need to know~$k$ in advance.
This allows us to run \findMeth{} and \infDef{} on all graphs, without having to specify a value for~$k$.
Instead, we set~$k=0$ for all instances.
In the final instance~$(G',k')$ computed from~$(G,k=0)$ we will have~$\tau(G') - k' = \tau(G)$, where $\tau(G)$ denotes the vertex cover number of $G$.
For graphs~$G'$ which become empty, $-k'$ is the vertex cover number of the original graph~$G$.

We set a maximum recursion limit for \findMeth{} such that only sequences of at most two or three rules are found.
We will refer to \texttt{FAR2} and \texttt{FAR3} as the \findRed{} method which only searches for sequences of at most two or three forward and backward rules, respectively.

We used inflation ratios~$\alpha$ equal to 10, 20 and 50 percent, and we will refer to the different \infDef{} configurations as \texttt{ID10}, \texttt{ID20}, and \texttt{ID50}, respectively.
We also tested our \localInfDef{} method with~$\alpha = 20\%$, which we have found to work best in preliminary experiments, which we will refer to as \texttt{LID20}.

All these configurations were tested on the three datasets with a maximum running time of one hour.

\subsection{Results}\label{ssec:exp-results}

{
\subparagraph*{Confluence.}
Confluence can be proved using for example conflict pair analysis \cite{Ehrig14}, which is not straightforward and does not work for all types of data reduction rules.
However, disproving confluence is potentially much easier, as it suffices to find one example where applying the rules in different order yields different instances.

For each pair of forward rules in \cref{tab:forward_glossary}, we tested on the set of all graphs with at most 9 vertices\footnote{We obtained these graphs from Brendan McKay's website \url{https://users.cecs.anu.edu.au/~bdm/data/graphs.html}} whether we can obtain irreducible, but non-isomorphic graphs by randomly applying the two rules to the same graph.
More precisely, we test whether a set~$\mathcal{R} = \{R_a, R_b, \text{Deg0}\}$ is confluent for two forward rules~$R_a$ and~$R_b$.
We include the Degree-0 Rule in these sets, because a difference in the number of isolated vertices is only a minor detail which we do not wish to take into account. 
We note that, the inclusion of the Degree-0 Rule in these sets was never the reason that a set was not confluent.

Our results are summarized in \cref{fig:confluence_matrix}.
\begin{figure}[t!]
	\centering
	\includegraphics[width=0.49\textwidth]{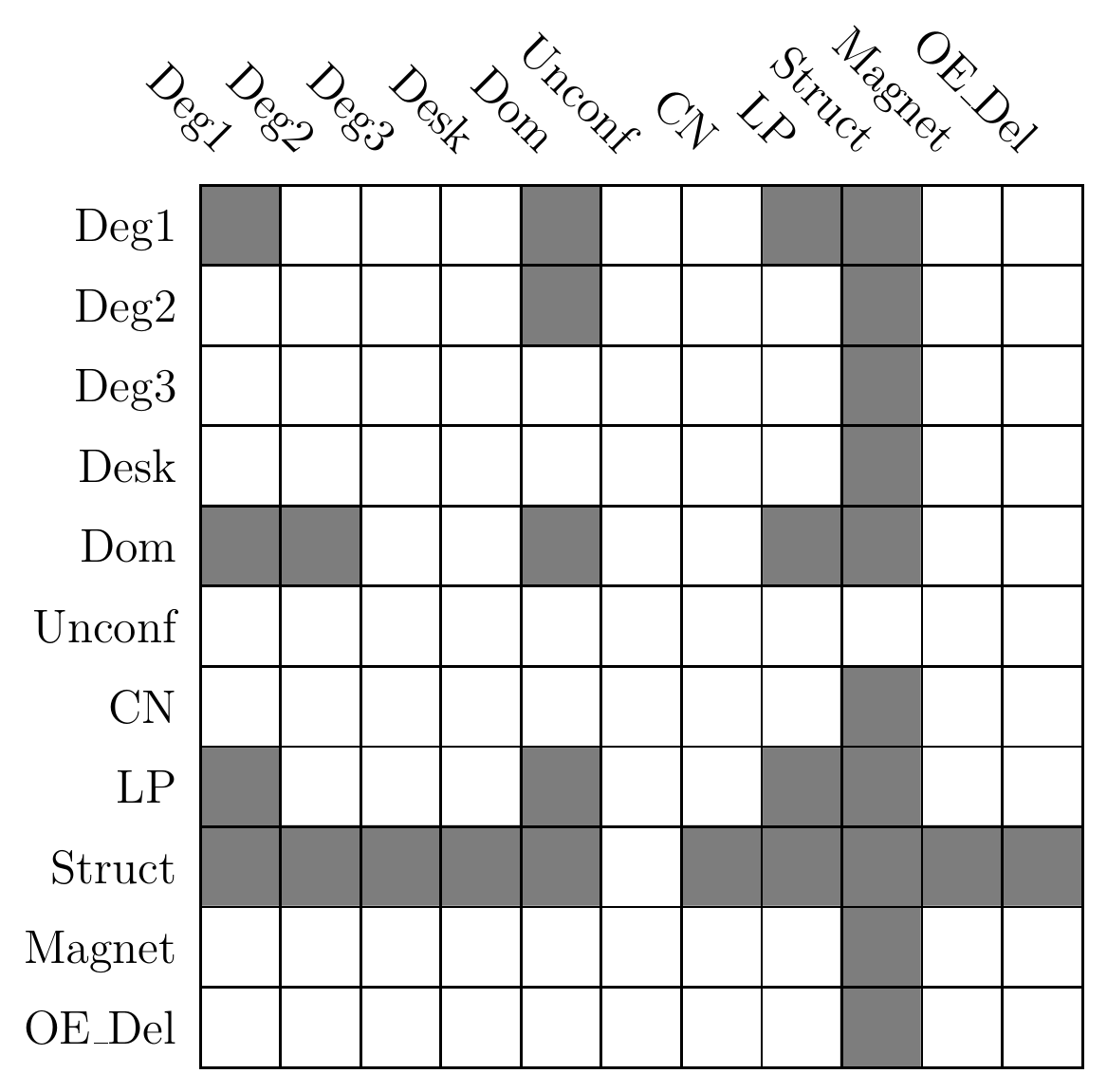}
	\caption{A matrix depicting which pairs of forward rules we have found to be non-confluent. A cell corresponding to a row rule~$R_a$ and a column rule~$R_b$ is colored white if the set~$\mathcal{R} = \{R_a, R_b, \text{Deg0}\}$ is not confluent. The Degree-0 rule is always included to not take into account differences in number of isolated vertices left after applying the rules. Gray cells correspond to sets which may be confluent, meaning that no counter-example was found for them.}
	\label{fig:confluence_matrix}
\end{figure}
The figure clearly shows that most pairs of forward rules are not confluent.
This means that the relative order of these rules may affect the final instance.
}

We demonstrate how much the \findRed{} and \infDef{} methods were able to shrink the irreducible DIMACS, SNAP, and PACE graphs, which were obtained by exhaustively applying a set of forward rules.
The results are summarized in \cref{fig:instance_shrinkage,tab:instance_shrinkage}.
For an example of a sequence of data reduction rules found with the \findMeth{} see \cref{fig:found-sequence-ex} in the appendix.
\begin{figure}[t!]
	\centering
	\includegraphics{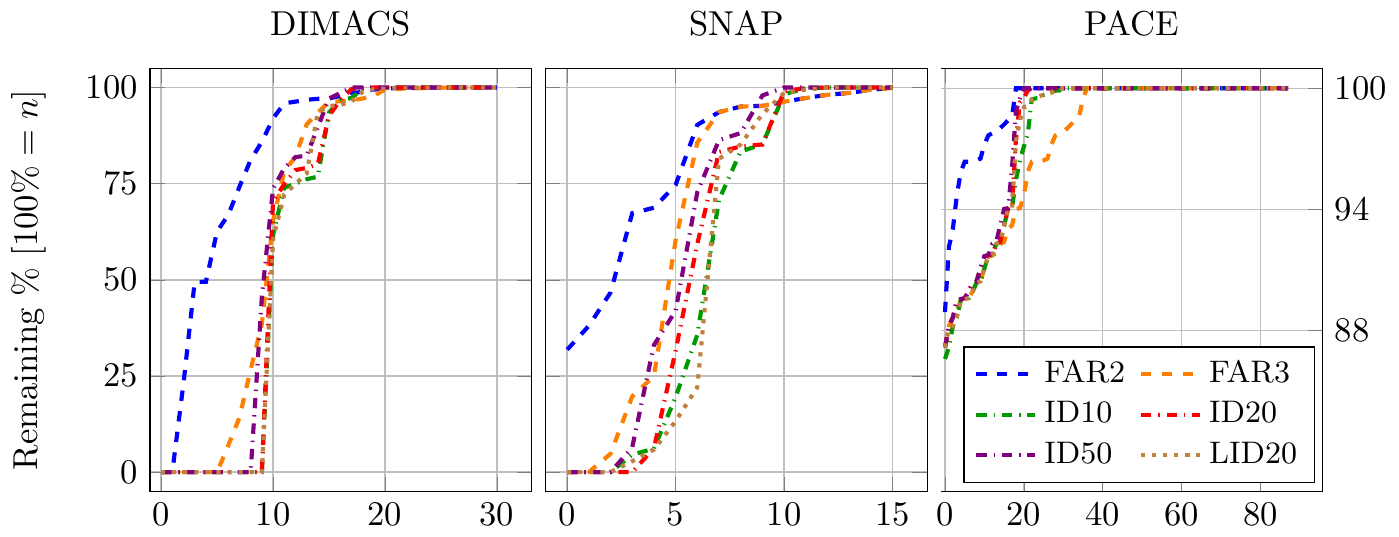}
	\caption{
		Cactus plots depicting how many irreducible instances were shrunk to a given fraction of their size (measured in vertices). 
		The order of the instances is chosen for each configuration such that the size fractions of the instances are increasing. 
		Note that the y-axis for DIMACS and SNAP start at zero (i.\,e.\ instances reduced to the empty graph); the y-axis for PACE does \emph{not}. 
	}
	\label{fig:instance_shrinkage}
\end{figure}

\begin{table}[t!]
	\caption{Summery of the average relative size (in terms of number of vertices) achieved by each of the configurations on the three datasets. The best value per dataset is given in bold.}
	\label{tab:instance_shrinkage}
	\centering
	\begin{tabularx}{0.9\textwidth}{rXXXXXX}
		\toprule
				& \texttt{FAR2}		& \texttt{FAR3}		& \texttt{ID10}		& \texttt{ID20}		& \texttt{ID50}		& \texttt{LID20} \\
		\midrule
		DIMACS 	& 82.6\% 			& 67.0\% 			& \textbf{62.9}\%	& 63.6\% 			& 66.1\% 			& 63.4\% \\
		SNAP 	& 80.6\% 			& 66.8\% 			& \textbf{56.4}\%	& 59.3\%			& 64.1\% 			& 56.4\% \\
		PACE 	& 99.2\% 			& \textbf{97.4}\%	& 97.8\% 			& 98.1\% 			& 98.1\% 			& 98.0\% \\
		\bottomrule
	\end{tabularx}
\end{table}

Notably, ten DIMACS and four SNAP instances were reduced to an empty graph by the \infDef{} methods.
Five further DIMACS graphs shrank to around 80\% of their size, and half of the DIMACS graphs did not really shrink at all.
We conclude that the \infDef{} approach seems to either work really well or nearly not at all for a given instance.

On average the \texttt{ID10} configuration produced the smallest irreducible instances, see \cref{tab:instance_shrinkage}.
From the \texttt{FAR2} configuration it can be seen that already applying only two forward/backward rules in a sequence can considerably reduce the size of some graphs.
In this case, the first rule is always a backward rule, and the second always a forward rule.
However, using up to three rules in a sequence gave significantly better results.

For \infDef{}, we see that small inflation ratios~$\alpha$ around 10\% perform best on average.
Increasing the inflation ratio leads to slightly worse results, especially for the SNAP instances.

Next, in \cref{fig:compression_vs_degree}, we see that the \texttt{ID10} method was able to shrink graphs to empty graphs mostly for graphs with the lowest average degree (8--14), with the exception of one instance with an average degree of around 45.
\begin{figure}[t!]
	\includegraphics{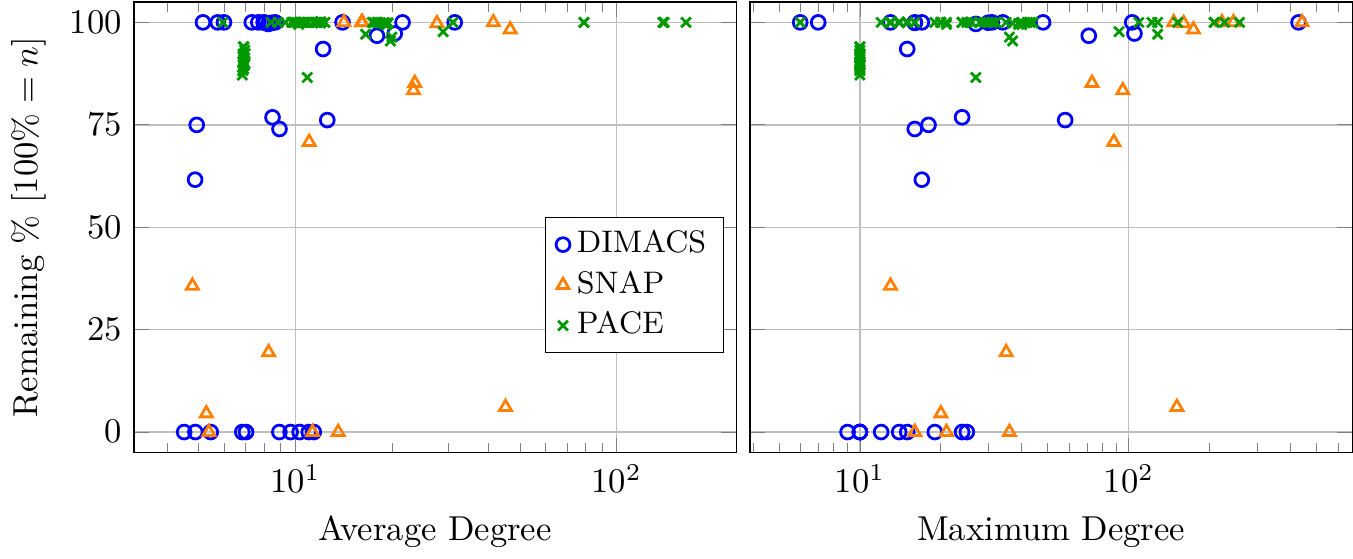}
	\caption{Scatter plot relating the relative size (measured in vertices) of the shrunk instances to the average and maximum degree of these graphs for the \texttt{ID10} configuration. Note that a logarithmic scale is used for the x-axis.}
	\label{fig:compression_vs_degree}
\end{figure}
However, a large number of graphs with an average degree of 8--14 were not shrunk considerably.
The other configurations, namely \texttt{ID20}, \texttt{ID50}, and \texttt{FAR3} exhibit the same behavior.
Similar behavior is also observed by replacing the average degree with the maximum degree.
For the most part, only graphs with a relatively small maximum degree were able to be shrunk considerably.
We conclude that our approach is most viable on sparse irreducible graphs.

In \cref{fig:shrinkage_over_time}, we show how the graph size changes over time with \findRed{} and \infDef{} on two example graphs.
\begin{figure}[t!]
	\centering
	\includegraphics{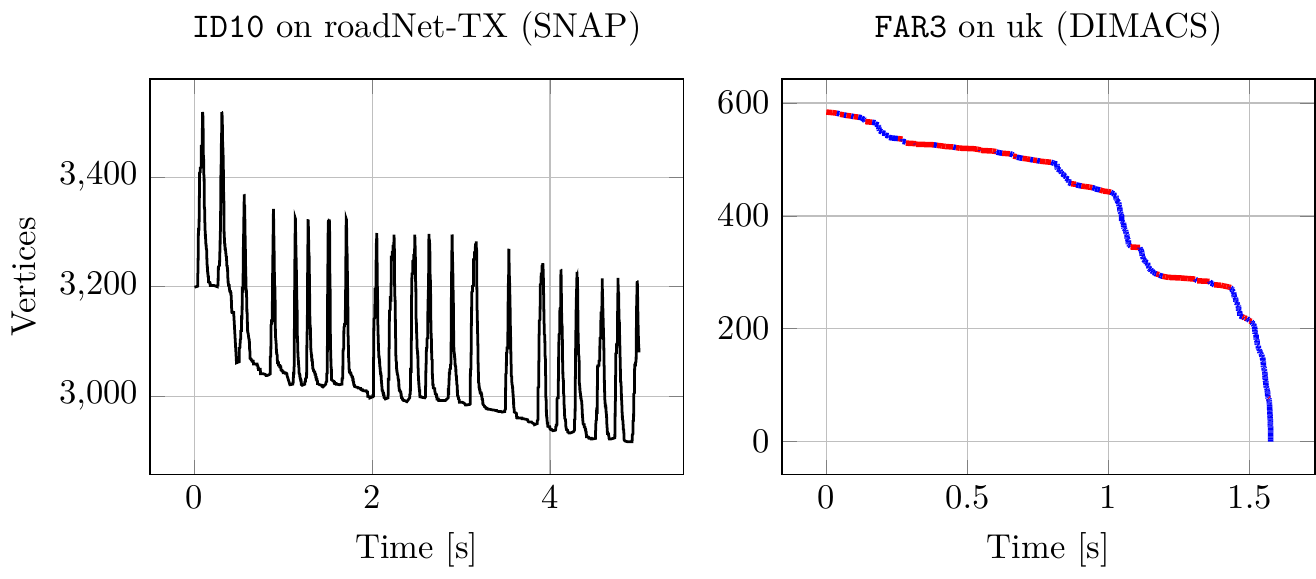}
	\caption{
		The shrinkage of two instances by ID10 and FAR3 configurations on already kernelized instances. 
		The~$x$ axis is time, and~$y$ is number of vertices. 
		In the right figure red segments mean that the graph was shrunk by using a sequence of at least two forward and backward rules, whereas the blue segments indicate the use of forward rules only.}
	\label{fig:shrinkage_over_time}
\end{figure}
It can be clearly observed how \texttt{ID10} repeatedly increases the number of vertices, which is the inflation phase, and then subsequently reduces it, which is the deflation phase.
A slow downward trend of the number of vertices can be observed.

For the \texttt{FAR3} configuration it can be seen that there are phases in which only forward rules have to be used to shrink the graph, and phases where longer sequences of forward and backward rules are needed.
Sometimes a sudden large decrease in the number of vertices using only forward rules can be observed, which one may think of as a cascading effect.
The graph only had to be changed by a small amount, triggering a cascade of forward rules.

\section{Conclusion}

Our work showed the large potential in the general idea of undoing data reduction rules to further shrink instances that are irreducible with respect to these rules.
While the results for some instances are very promising, our experiments also revealed that other instances resist our attempts of shrinking them through preprocessing. 
From a theory point of view this is no surprise, as we deal with NP-hard problems after all.
However, there is a lot of work that still can be done in this direction. 
Similar to the branch\&bound approach, many clever heuristic tricks will be needed to find solutions in the vast search space.
Such heuristics might, for example, employ machine learning to guide the search.
As mentioned before, our approach is not limited to \VC{}.
Looking at other problems is future work though.
A framework for applying our approach on graph problems could be another next step.
Also, our approach should be easily parallelizable.

Besides all these practical questions, there are also clear theoretical challenges:
For example, for a given set of data reduction rules is there \emph{always} (for each possible instance) a sequence of backwards and forward rules to obtain an equivalent instance of constant size?
Note that this would not contradict the NP-hardness of the problems: 
Such sequences are probably hard to find and could even be of exponential length.

\bibliographystyle{plainnat} 
\bibliography{extracted}

\newpage

\appendix

\section{Data Set \& Illustration}\label{chap:tables}

\pgfplotstableread[col sep = comma]{material/dimacs_big_table.csv}\data %
\begin{table}[h!]
    \caption{\normalfont 
		DIMACS instance properties. 
		$n,m$ - the number of vertices and edges in the original instance, 
		$n'_{\texttt{AIR}}$ - number of vertices after applying reductions used by \citet{AI16}, 
		$n',m'$ - number of vertices and edges after applying forward rules.
		$n''_\texttt{ID10},m''_\texttt{ID10}$ - number of vertices and edges in the smaller irreducible instance found by $\texttt{ID10}$. 
		}
	\small
\pgfplotstabletypeset[columns={instance,n,m,nwata,nstandard,mstandard,nid10,mid10},%
	columns/instance/.style={string type,column name=Graph,column type = {r}},
	columns/n/.style={column name=$n$,precision=1,column type = {r}},
	columns/m/.style={column name=$m$,precision=1,column type = {r}},
	columns/nwata/.style={column name=$n'_{\texttt{AIR}}$,precision=1,column type = {r}},
	columns/nstandard/.style={column name=$n'$,precision=1,column type = {r}},
	columns/mstandard/.style={column name=$m'$,precision=1,column type = {r}},
	columns/nid10/.style={column name=$n''_{\texttt{ID10}}$,precision=1,column type = {r}},
	columns/mid10/.style={column name=$m''_{\texttt{ID10}}$,precision=1,column type = {r}},
	every head row/.style ={before row=\toprule, after row=\midrule},
    every last row/.style ={after row=\bottomrule}]{\data}
\end{table}

\pgfplotstableread[col sep = comma]{material/snap_big_table.csv}\data %
\begin{table}[t]
    \caption{\normalfont SNAP instance properties. 
		$n,m$ - the number of vertices and edges in the original instance, 
		$n'_{\texttt{AIR}}$ - number of vertices after applying reductions used by \citet{AI16}, 
		$n',m'$ - number of vertices and edges after applying forward rules.
		$n''_\texttt{ID10},m''_\texttt{ID10}$ - number of vertices and edges in the smaller irreducible instance found by $\texttt{ID10}$. 
		On one graph we were unable to run the program by \citet{AI16}. The number of vertices was marked with -1 in this one case.
    }
	\small
\pgfplotstabletypeset[columns={instance,n,m,nwata,nstandard,mstandard,nid10,mid10},%
	columns/instance/.style={string type,column name=Graph,column type = {r}},
	columns/n/.style={column name=$n$,precision=1,column type = {r}},
	columns/m/.style={column name=$m$,precision=1,column type = {r}},
	columns/nwata/.style={column name=$n'_{\texttt{AIR}}$,precision=1,column type = {r}},
	columns/nstandard/.style={column name=$n'$,precision=1,column type = {r}},
	columns/mstandard/.style={column name=$m'$,precision=1,column type = {r}},
	columns/nid10/.style={column name=$n''_{\texttt{ID10}}$,precision=1,column type = {r}},
	columns/mid10/.style={column name=$m''_{\texttt{ID10}}$,precision=1,column type = {r}},
	every head row/.style ={before row=\toprule, after row=\midrule},
    every last row/.style ={after row=\bottomrule}]{\data}
\end{table}

\pgfplotstableread[col sep = comma]{material/pace_big_table_short.csv}\data %
\begin{table}[t]
    \caption{\normalfont PACE instance properties. 
		$n,m$ - the number of vertices and edges in the original instance, 
		$n'_{\texttt{AIR}}$ - number of vertices after applying reductions used by \citet{AI16}, 
		$n',m'$ - number of vertices and edges after applying forward rules.
		$n''_\texttt{ID10},m''_\texttt{ID10}$ - number of vertices and edges in the smaller irreducible instance found by $\texttt{ID10}$. 
    Only a third of the dataset is included.}
	\small
\pgfplotstabletypeset[columns={instance,n,m,nwata,nstandard,mstandard,nid10,mid10},%
	columns/instance/.style={string type,column name=Graph,column type = {r}},
	columns/n/.style={column name=$n$,precision=1,column type = {r}},
	columns/m/.style={column name=$m$,precision=1,column type = {r}},
	columns/nwata/.style={column name=$n'_{\texttt{AIR}}$,precision=1,column type = {r}},
	columns/nstandard/.style={column name=$n'$,precision=1,column type = {r}},
	columns/mstandard/.style={column name=$m'$,precision=1,column type = {r}},
	columns/nid10/.style={column name=$n''_{\texttt{ID10}}$,precision=1,column type = {r}},
	columns/mid10/.style={column name=$m''_{\texttt{ID10}}$,precision=1,column type = {r}},
	every head row/.style ={before row=\toprule, after row=\midrule},
    every last row/.style ={after row=\bottomrule}]{\data}
\end{table}

\begin{figure}[t!]
	\centering
	\includegraphics[scale=0.4]{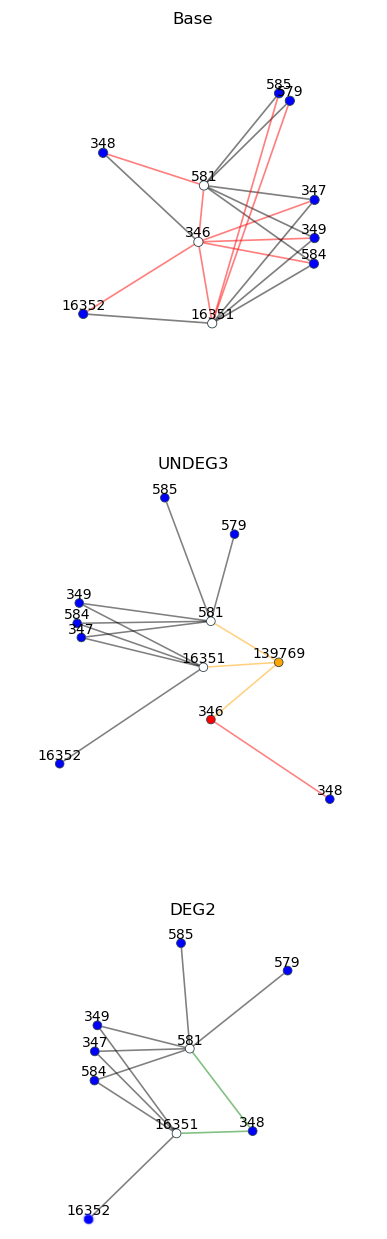}
	\caption{A visualization of a graph modification corresponding to the sequence (Undeg3,Deg2) of forward and backward rules, which in this case shrink a graph. The visualization is automatically generated from a log file of our \findMeth{} Method. Starting with the graph at the top, the smaller graph on the bottom is obtained. Over each graph except the first, the name of the rule which was applied to the previous graph is written. The blue vertices are the boundary vertices of the graph modification. The red edges and vertices which are deleted in the next layer, whereas green vertices and edges were new compared to the previous layer. A yellow vertex or edge is new, and will be deleted in the next layer.}
	\label{fig:found-sequence-ex}
\end{figure}

\clearpage

\section{Compendium of Data Reduction Rules for \VC{}}\label{chap:vc_rrs}

In this section, we aim to provide an overview of data reduction rules (forward rules) for \VC{} in \cref{sect:forward_rules} and show ways how some of them may be undone and formulate backward rules based on them in \cref{sect:backward_rules}.

Many kernelization results for \VC{} are already summarized in a recent survey paper by \citet{FJK18}.
They focus mainly on data reduction rules for kernelizing or removing low-degree vertices.
However, some data reduction rules that are used in practical implementations are not covered in their survey.
This includes, for example, the Unconfined Rule~\cite{XN13} (\cref{rr:unconfined}) used by \citet{AI16} in their \VC{} solver implementation.
The reduction rule suite of \citet{AI16} was later used by, for example, \citet{HLS20} in their winning \VC{} implementation for the 2019 PACE Challenge~\cite{PACE}.

Another interesting rule is, for example, the Struction Rule~\cite{EHW84} (\cref{rr:struction}).
The rule has received repeated interest over the years  \cite{GLS21,AHL03,Loz17,CKX10}), however only few practical experiments have been made with it.
These experiments were mostly performed on random graphs and in combination with only few other data reduction rules \cite{EHW84,AHL03}.
However, \citet{GLS21} developed a modification of the Struction Rule for the \textsc{Maximum Weight Independent Set} problem, where the rule showed promising results, in that many graphs were reduced to an empty graph because of it.
For this reason, we revisit the original Struction Rule.

In our overview we also include other, perhaps promising rules such as the Magnet and Edge Deletion rules (\cref{rr:magnet,rr:oe_delete}).
Additionally, we show that the 2-Clique Neighborhood Rule by \citet{FJK18} can also be applied to vertices of arbitrary degree in polynomial-time and we provide a generalization of the Unconfined Rule (\cref{rr:unconfinedpp}).

We remark that all listed rules with references are safe even if not explicitly stated.
Additionally, a small technicality is that the parameter~$k$ in \VC{} is defined as a natural number.
However, the application of some rules may produce negative values for~$k$.
We will also allow~$k$ to become negative, because if it does then it is clearly a no-instance.

\subsection{Forward rules}\label{sect:forward_rules}
We begin our list of forward rules (data reduction rules) starting with perhaps some of the most simple rules.
These reduction rules are typically found in many solvers, due to their simplicity and efficient implementation.

The most simple rule is the Degree-0 Rule, which simply removes isolated vertices from the graph.
\begin{frule}[Degree-0 \cite{BG93}]
\label{rr:deg0}
	Let~$v$ be a degree-0 vertex. Then remove the vertex~$v$ from the graph.
\end{frule}
Clearly, no minimum vertex cover can contain isolated vertices, therefore it is safe to delete them.

We proceed with the Degree-1 Rule, which as the name suggests is used to remove degree-1 vertices from the graph.

\begin{frule}[Degree-1 \cite{BFR98}]
\label{rr:deg1}
	Let~$v$ be a degree-1 vertex and~$u$ its unique neighbor. Then delete~$u,v$ from~$G$ and decrease~$k$ by one.
\begin{figure}[t]
	\centering
	\begin{subfigure}[c]{.45\textwidth}
		\centering
		\includegraphics{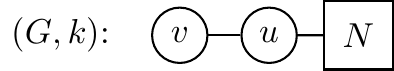}
	\end{subfigure}
	\begin{subfigure}[c]{.45\textwidth}
		\centering
		\includegraphics{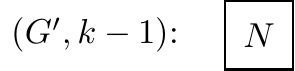}
	\end{subfigure}
\caption{An illustraion of \cref{rr:deg1}. Note that we use square nodes to represent whole sets of vertices.}
\label{fig:rr_deg1}
\end{figure}
\end{frule}
For an illustration of \cref{rr:deg1} see \cref{fig:rr_deg1}.

The next rule is the Degree-2 Folding Rule, which deals with degree-2 vertices whose neighbors are not adjacent.
This rule allows to \emph{merge} the degree-2 vertex with its two neighbors.
By merging a set of vertices~$S$ we mean creating a new vertex~$v'$, adding all edges between~$c$ and~$N_G(S)$, and then deleting~$S$ from the graph.
For an illustration of the rule see \cref{fig:rr_deg2}.

\begin{frule}[Degree-2 Folding \cite{SF99}]
\label{rr:deg2}
	Let~$v$ be a degree-2 vertex and~$a,b$ its two neighbors. If~$a,b$ are nonadjacent, then merge~$v,a,b$ into a single vertex and decrease~$k$ by one.
\begin{figure}[t]
	\centering
	\begin{subfigure}[c]{.45\textwidth}
		\centering
		\includegraphics{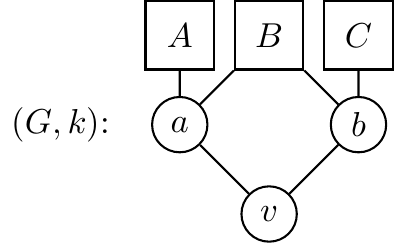}
	\end{subfigure}
	\begin{subfigure}[c]{.45\textwidth}
		\centering
		\includegraphics{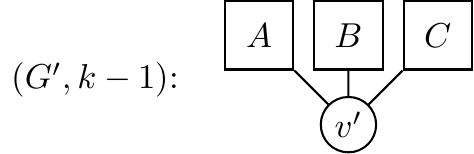}
	\end{subfigure}
\caption{An illustration for the Degree-2 folding rule (\cref{rr:deg2}). The vertices~$v,a,b$ are merged into the single vertex~$v'$.}
\label{fig:rr_deg2}
\end{figure}
\end{frule}

Similarly to the Degree-2 Folding Rule there exists a rule for handling degree-3 vertices.
However, the rule is only applicable to vertices whose neighborhood is an independent set.
For an example see \cref{fig:rr_deg3}.

\begin{frule}[Degree-3 Independent Set \cite{SF99}]
\label{rr:deg3}
	Let~$v$ be a degree-3 vertex and~$N(v) = \{a, b, c\}$ an independent set.
	Then
	\begin{itemize}
	\item remove~$v$,
	\item add the edges~$\{a, b\}, \{b, c\}$,
	\item add the edges from~$\{\{a, x\} \mid x \in N_G(b)\} \cup \{\{b, x\} \mid x \in N_G(c)\} \cup \{\{c, x\} \mid x \in N_G(a)\}$ if they do not already exist.
	\end{itemize}
\begin{figure}[t]
	\centering
	\begin{subfigure}[c]{.45\textwidth}
		\centering
		\includegraphics{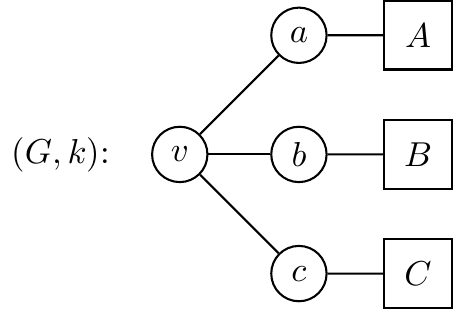}
	\end{subfigure}
	\begin{subfigure}[c]{.45\textwidth}
		\centering
		\includegraphics{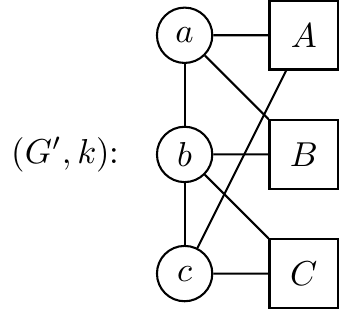}
	\end{subfigure}
\caption{An illustration for the Degree-3 Independent Set Rule (\cref{rr:deg3}).}
\label{fig:rr_deg3}
\end{figure}
\end{frule}
We note that in \cref{rr:deg3} the six possible permutations of~$\{a, b, c\}$ can result in different outcomes for the rule, meaning that the resulting graphs may not be isomorphic.
However, the number of newly inserted edges is unaffected by the permutation.

High-degree vertices may also be removed thanks to the following rule by \citet{BG93}.
This rule is also the last degree-based rule in this overview.
\begin{frule}[Degree~$> k$ \cite{BG93}]
\label{rr:degk}
	Let~$v$ be a vertex of degree greater than~$k$. Then delete~$v$ and decrease~$k$ by one.
\end{frule}
Applying \cref{rr:deg0} and \cref{rr:degk} rules exhaustively results in the known Buss kernel \cite{BG93}.
One can show that after applying the rules if the graph has more than~$k^2+k$ vertices or~$k^2$ edges, then the instance is a no-instance \cite{BG93}.

What will now follow is a series of more general data reduction rules, meaning that they are not restricted to vertices of specific degrees.
Perhaps the simplest of these is the Domination Rule\footnote{Also known as neighborhood or dominance reduction \cite{AI16,Loz17}.}.

We say a vertex~$u$ dominates another vertex~$v$ if~$N[v] \subseteq N[u]$.
For an illustration of the rule see \cref{fig:rr_dom}.

\begin{frule}[Domination \cite{SF99}]
\label{rr:domination}
	Let~$u,v$ be two adjacent vertices such that~$u$ dominates~$v$, i.\,e., $N[v] \subseteq N[u]$. Then delete~$u$ and decrease~$k$ by one.
	\begin{figure}[t]
		\centering
		\begin{subfigure}[c]{.45\textwidth}
			\centering
			\includegraphics{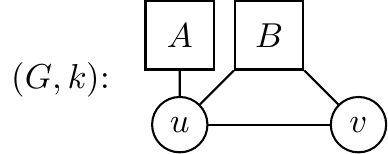}
		\end{subfigure}
		\begin{subfigure}[c]{.45\textwidth}
			\centering
			\includegraphics{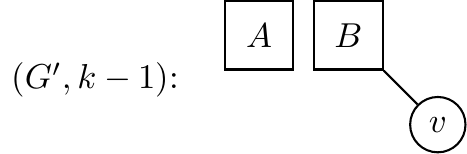}
		\end{subfigure}
	\caption{An illustration for \cref{rr:domination}. The vertex~$u$ dominates~$v$ in~$G$.}
	\label{fig:rr_dom}
	\end{figure}
\end{frule}
The Domination Rule can be thought of a generalization of the Triangle Rule we have seen in \cref{sec:intro}.

The remaining reduction rules are more complex, both in terms of their formulation and difficulty of implementation.
The Unconfined Rule has perhaps the most unique formulation out of all rules presented here, as it is directly an algorithm.
This rule can be thought of as a generalization of the Domination Rule, where an additional step is taken if a vertex~$v$ nearly dominates a vertex~$u$.

\begin{frule}[Unconfined \cite{XN13}, reformulation by \citet{AI16}]
	\label{rr:unconfined}
	Let~$v$ be a vertex for which \cref{alg:unconf} returns yes. 
	\begin{algorithm}
		\caption{Unconfined Check.}
		\label{alg:unconf}
		Let~$S = \{v\}$\;
		Find~$u \in N(S)$ such that~$|N(u) \cap S| = 1$ and~$|N(u) \setminus N[S]|$ is minimized.\;
		\lIf{there is no such vertex}{\textbf{return no}}
		\lIf{$N(u) \setminus N[S] = \emptyset$}{\textbf{return yes}}
		\lIf{$N(u) \setminus N[S] = \{w\}$}{add~$w$ to~$S$ and goto line 2}
	\end{algorithm}
	Then delete~$v$ and decrease~$k$ by one.
	\begin{figure}[t]
		\centering
		\includegraphics{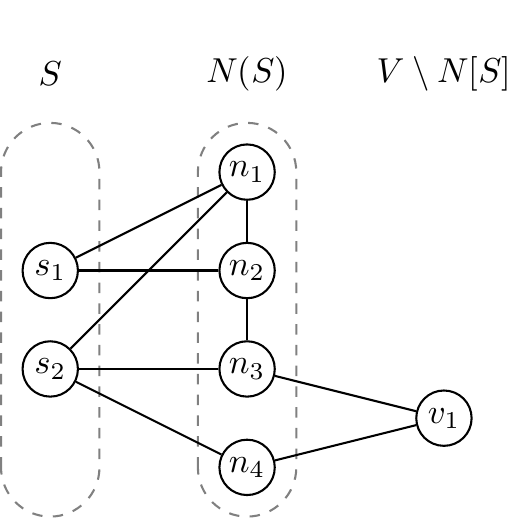}
		\caption{An example for performing a single iteration of \cref{rr:unconfined}. Assume the set~$S$ was initially equal to~$\{s_1\}$ and we want to see if~$s_1$ is unconfined. Given the depicted situation, where one iteration was already performed (by adding~$s_2$ to~$S$),  one can directly conclude that~$s_1$ is unconfined due to~$n_2$ having a single neighbor in~$S$ and no neighbors outside of~$N[S]$. If~$n_2$ did not exist, then one should instead add~$v_1$ to~$S$. Note that the vertices in~$V \setminus N[S]$ may also have neighbors outside of~$N[S]$, which have no impact on the outcome of a single iteration of the unconfined rule.}
		\label{fig:rr_unconf}
	\end{figure}
\end{frule}
Vertices for which the algorithm in \cref{rr:unconfined} returns yes are called \emph{unconfined}.
For an example of a case where the Unconfined Rule can be applied see \cref{fig:rr_unconf}.

\citet{AI16} implemented a generalization of \cref{rr:unconfined}, which in their source code is called the Diamond Reduction (which searches for a~$K_{2,2}$), but is not discussed in their paper directly.

We suggest a further generalization: the Unconfined-$\kappa$ Rule.
This rule takes as additional input a small positive constant number~$\kappa$ which we use to guarantee polynomial running time.
Setting~$\kappa$ to a large constant may allow the rule to be executed in more cases, however potentially at the cost of a higher running time.
\begin{frule}[Unconfined-$\kappa$]
\label{rr:unconfinedpp}
	Let~$v$ be a vertex for which \cref{alg:unconf-gen} returns yes. 
	\begin{algorithm}[t!]
		\caption{Unconfined-$\kappa$ Check.}
		\label{alg:unconf-gen}
		Let~$S = \{v\}$\;
		Find~$X \subseteq N(S)$ with~$0 < |X| \leq \kappa$ such that~$|N(X) \setminus N[S]| \leq 1$ and for~$Y = N(X) \cap S$ it holds that~$|Y| = |X|$ and~$G[X,Y]$ is a complete bipartite graph.\;
		\lIf{there is no such set~$X$}{\textbf{return no}}
		\lIf{$N(X) \setminus N[S] = \emptyset$}{\textbf{return yes}}
		\lIf{$N(X) \setminus N[S] = \{w\}$}{add~$w$ to~$S$ and goto line 2}
	\end{algorithm}
	Then delete~$v$ and decrease~$k$ by one.
\end{frule}
The safeness proof of \cref{rr:unconfinedpp} is analogous to that of \cref{rr:unconfined} by \citet{XN13}.
We start with the strong (and possibly incorrect) assumption that no minimum vertex cover of~$G$ contains~$v$.
The procedure behind \cref{rr:unconfinedpp} tries to derive a contradiction from this sole assumption in order to prove that there exists a minimum vertex cover which contains~$v$.
The algorithm maintains the invariant that no minimum vertex cover of~$G$ contains~$S$ and all minimum vertex covers of~$G$ contain~$N(S)$.
This rule essentially derives a proof by contradiction for the fact that the vertex~$v$ is contained in at least one minimum vertex cover of~$G$.

The following lemma is the main idea behind the Unconfined-$\kappa$ Rule.
We will denote by~$\vc(G)$ the set of all minimum vertex covers of~$G$.
\begin{lemma}
\label{lem:unconfinedpp}
	Given a graph~$G$ and a set~$S \subseteq V$ such that for all~$C \in \vc(G)$ it holds~$S \cap C = \emptyset$.
	Then for every~$X \subseteq N(S)$ and corresponding~$Y = N(X) \cap S$ such that~$|Y| = |X|$ and~$G[X,Y]$ is a complete bipartite graph, it must hold that for all~$C \in \vc(G)$~$N(X) \setminus N[S] \not \subseteq C$.
\end{lemma}
\begin{proof}
	Let~$G, S, X, Y$ be as above.
	Assume that there exists a~$C \in vc(G)$ such that~$N(X) \setminus N[S] \subseteq C$.
	Due to~$S \cap C = \emptyset$ it holds~$N(S) \subseteq C$ and in particular~$X \subseteq C$.
	Then~$C' = (C \setminus X) \cup Y$ is a minimum vertex cover.
	However, $S \cap C' \neq \emptyset$, a contradiction to the definition of~$S$.
\end{proof}
We are now ready to prove the safeness of \cref{rr:unconfinedpp}.
\begin{lemma}
\cref{rr:unconfinedpp} is safe.
\end{lemma}
\begin{proof}
We show that if the algorithm described in \cref{rr:unconfinedpp} returns yes, then the vertex~$v$ is contained in at least one minimum vertex cover of~$G$, which then immediately implies the lemma. 

Assume no minimum vertex cover of~$G$ contains~$v$.
The algorithm initially sets~$S = \{v\}$.

First, we show that the algorithm maintains the invariant that no vertex in~$S$ is contained in any minimum vertex cover of~$G$.
Initially this is true, based on the assumption that was made at the start.
A new vertex~$w$ may only be added to~$S$ in Line 5.
The sets~$X$ and~$Y$ found in Line 2, shortly before executing Line 5 satisfy the conditions of \cref{lem:unconfinedpp}.
The lemma, in turn, directly implies that no minimum vertex cover of~$G$ contains~$w$.
Therefore it is safe to add it to~$S$.

Next we show that if the algorithm returns yes, then there exists a minimum vertex cover of~$G$ that contains~$v$.
The algorithm returns yes only in Line 4.
The sets~$X$ and~$Y$ found in Line 2, shortly before executing Line 4 satisfy the conditions of \cref{lem:unconfinedpp}.
Due to~$N(X) \setminus N[S] = \emptyset$, this is a contradiction to \cref{lem:unconfinedpp}, because~$\emptyset \subseteq C$ for all~$C \in \vc(G)$ (note that~$\vc(G)$ is never empty).

Because all steps were correct, this means that the assumption that no minimum vertex cover of~$G$ contains~$v$ was incorrect.
Therefore there exists a minimum vertex cover of~$G$ that contains~$v$.
\end{proof}

The next rule is based on the concept of alternative sets by \citet{XN13}.
Two disjoint independent sets~$A,B \subseteq V$ with~$|A| = |B| > 0$ are called \emph{alternative} if there exists a minimum vertex cover~$C$ of~$G$ such that~$C \cap (A \cup B) =$~$A$ or~$B$. A chordless 4-cycle, is a induced cycle on four vertices.
\begin{frule}[Desk \cite{XN13}]
\label{rr:desk}
	Let~$u_1u_2u_3u_4$ be a chordless 4-cycle.
	Let~$A = \{u_1, u_3\}, B = \{u_2, u_4\}$.
	If~$N(A) \cap N(B) = \emptyset$ and~$|N(A) \setminus B|, |N(B) \setminus A| \leq 2$,
	then remove~$A$ and~$B$ from~$G$ and for every nonadjacent~$a \in N(A) \setminus B$ and~$b \in N(B) \setminus A$ add the edge~$\{a, b\}$.
	Finally, decrease~$k$ by two.
	\begin{figure}[t]
		\centering
		\begin{subfigure}[c]{.45\textwidth}
		\includegraphics{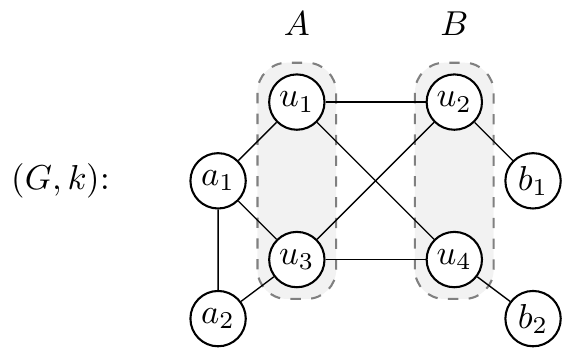}
		\end{subfigure}
		\begin{subfigure}[c]{.45\textwidth}
		\includegraphics{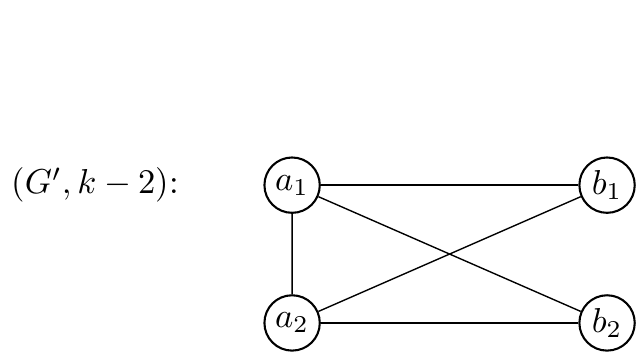}
		\end{subfigure}
		\caption{An example of the desk rule (\cref{rr:desk}).}
		\label{fig:rr_desk}
	\end{figure}
\end{frule}

The 2-Clique Neighborhood rule by \citet{FJK18} is applicable only to vertices with very dense neighborhoods.
Among others, it is required that the neighborhood can be partitioned into two cliques.

\begin{frule}[2-Clique Neighborhood \cite{FJK18}]
\label{rr:cn}
	Let~$v$ be a vertex such that there exists a partition~$(C_1, C_2)$ of~$N(v)$ with~$|C_1| \geq |C_2|$ and the following hold
	\begin{itemize}
	\item~$C_1$ and~$C_2$ are cliques in~$G$.
	\item Let~$M$ be the set of non-edges of~$G[N(v)]$. 
	For each~$c_1 \in C_1$, there is exactly one~$e \in M$ such that~$c_1 \in e$.
	\end{itemize}
	Then reduce~$(G, k)$ to~$(G', k-|C_2|)$, where~$G'$ is obtained from~$G$ by
	\begin{itemize}
	\item deleting~$v$ and~$C_2$, and
	\item for all~$\{c_1, c_2\} \in M$ with~$c_1 \in C_1, c_2 \in C_2$, add all missing edges between~$c_1$ and~$N_G(c_2)$.
	\end{itemize}
	\begin{figure}[t]
		\centering
		\begin{subfigure}[c]{.45\textwidth}
		\includegraphics{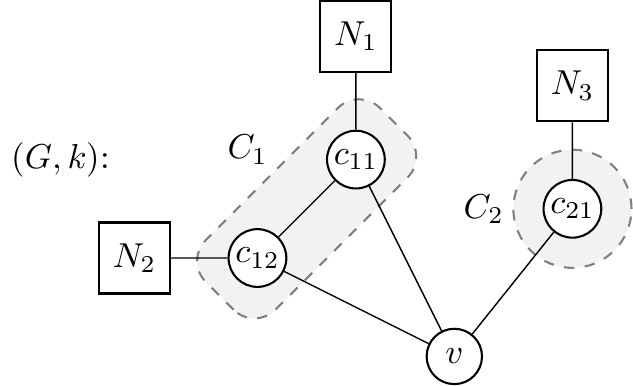}
		\end{subfigure}
		\begin{subfigure}[c]{.45\textwidth}
		\includegraphics{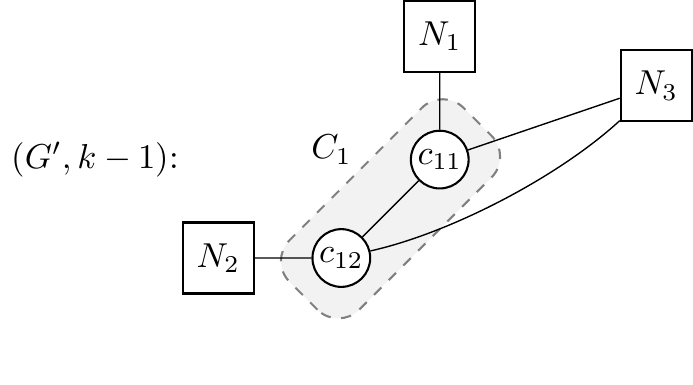}
		\end{subfigure}
		\caption{An example of the 2-clique neighborhood rule (\cref{rr:cn}).}
		\label{fig:rr_cn}
	\end{figure}
\end{frule}

Note that for \cref{rr:cn} it was not previously known in the literature, how to find the partitioning~$(C_1,C_2)$ of~$N(v)$ in polynomial-time if it exists.
\citet{FJK18} suggest applying the rule only to vertices~$v$ with~$\deg(v) \leq c$ where~$c$ is some constant.
\citet{AI16} instead use the Funnel Rule~\cite{XN13} which is a combination of repeated application of \cref{rr:domination} followed by a special case of \cref{rr:cn} where~$|C_2| = 1$.

We show that the partitioning needed for \cref{rr:cn} can indeed be found in polynomial-time if it exists.
\begin{lemma}
Given a vertex~$v \in V$, a partitioning of~$N(v)$ into two disjoint sets~$C_1, C_2$ fulfilling the requirements of \cref{rr:cn} can be found in polynomial-time, if it exists.
\end{lemma}
\begin{proof}
Consider the graph~$H = \overline{G[N(v)]}$.
Since~$C_1$ and~$C_2$ are required to be cliques in~$G$, they will be independent sets in~$H$.
Additionally, for each~$c_1 \in C_1$ we must have exactly one~$c_2 \in C_2$ such that~$\{c_1, c_2\} \notin E(G)$, meaning that~$\deg_{H}(c_1) = 1$.
This implies that if such a partitioning exists, then~$H$ must be a disjoint union of stars.

The partitioning~$(C_1, C_2)$ of~$N(v)$ can therefore be constructed as follows:

If~$H$ is not a disjoint union of stars, no such partitioning exists.
Otherwise, for each connected component~$H'$ of~$H$ pick a vertex~$u$ of maximum degree. Include~$u$ in~$C_2$ and~$N(u)$ in~$C_1$.
Note that a vertex of maximum degree may not be unique if the star is a single edge.
In this case the partitioning is also not unique\footnote{Different partitions will always yield isomorphic graphs. A vertex of maximum degree is only not unique if the star is an edge, say~$\{a, b\}$. If~$a$ in added to~$C_1$, then the 2-Clique Neighborhood Rule will connect~$a$ to~$N_G(b)$ and delete~$b$. Otherwise, the rule will delete~$a$ and connect~$b$ to~$N_G(a)$. Notice that in the two cases the vertices~$a$ and~$b$ correspond to each other (they have the same neighborhoods).}.
In the end, it only remains to check if~$|C_1| \geq |C_2|$.
If not, then no valid partitioning exists.

The above procedure clearly runs in polynomial-time.
It is simple to verify that a partitioning~$(C_1, C_2)$ constructed this way fulfills the requirements of \cref{rr:cn}.
Furthermore, if a valid partitioning exists the condition~$|C_1| \geq |C_2|$ will always be fulfilled, because the two sizes~$|C_1|$ and~$|C_2|$ are shared by all valid partitionings.
Consequently, a valid partitioning will always be found if it exists.
\end{proof}

The exhaustive application of the Degree-0, Degree-2 Folding, Degree-3 Independent Set, Domination and 2-Clique Neighborhood rules (\Cref{rr:deg0,rr:domination,rr:deg2,rr:cn,rr:deg3}) results in a graph with minimum degree at least 4 (for \cref{rr:cn} it is sufficient to only consider degree-3 vertices)~\cite{FJK18}.
We note that \citet{FJK18} also suggest reduction rules which remove degree-4 vertices in some cases.

A curious rule is the Optional Edge Deletion Rule by \citet{BHH85}.
It only allows to delete certain edges, but no vertices and the parameter~$k$ is unaffected.
The edge between two adjacent vertices~$a$ and~$b$ may be removed if there is third vertex~$c$ adjacent to either~$a$ or~$b$ and~$N(c) \subseteq N(a) \cup N(b)$.
The two vertices~$a$ and~$b$ can be thought of as jointly dominating the vertex~$c$.
For an example see \cref{fig:rr_oe_delete}.

\begin{frule}[Optional Edge Deletion \cite{BHH85}]
\label{rr:oe_delete}
Let~$a,b,c \in V$ be pairwise distinct vertices such that~$\{a, b\} \in E$, the vertex~$c$ is adjacent to exactly one of~$a$ or~$b$ and~$N(c) \subseteq N(a) \cup N(b)$.
Then delete the edge~$\{a, b\}$.
\begin{figure}[t]
	\centering
	\begin{subfigure}[c]{.45\textwidth}
		\centering
		\includegraphics{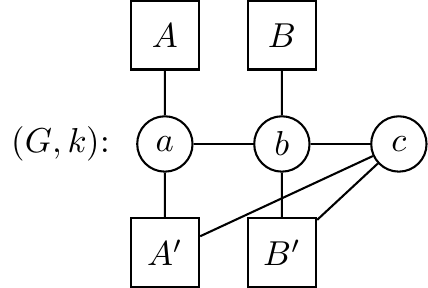}
	\end{subfigure}
	\begin{subfigure}[c]{.45\textwidth}
		\centering
		\includegraphics{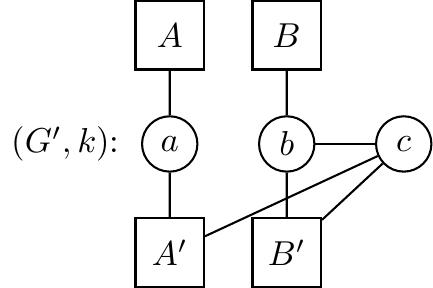}
	\end{subfigure}
\caption{An example for the application of \cref{rr:oe_delete}. Note that it is not required by the rule that~$A \cap B = \emptyset$ and~$A' \cap B' = \emptyset$.}
\label{fig:rr_oe_delete}
\end{figure}
\end{frule}

The safeness of the Optional Edge Deletion Rule is easy to verify.
If~$S$ is a vertex cover of size at most~$k$ for the original graph~$G$, then it is also a vertex cover for the reduced graph~$G'$.
For the other direction, let~$S'$ be a vertex cover of size at most~$k$ for~$G'$.
Assume without loss of generality that the vertex~$c$ is adjacent to~$b$.
If~$\{a, b\} \cap S' \neq \emptyset$, then~$S'$ is also a vertex cover for~$G$.
Otherwise, if~$\{a, b\} \cap S' = \emptyset$, then~$N_{G'}[c] \subseteq S'$, and hence~$S = (S' \cup \{b\}) \setminus \{c\}$ is a vertex cover for~$G$.

\subsubsection{Stability transformations}
The so-called stability transformations are reduction rules for \independentset{}.
The name comes from the fact that they change the stability number of a graph, which is the size of a maximum independent set, by a constant.

Reduction rules for \independentset{} can be reused also for \VC{}.
Given an instance~$(G,k)$ of \VC{} we can construct the equivalent instance~$(G, n-k)$ of \independentset{}.
The resulting \independentset{} instance may be reduced to~$(G', \alpha)$, which we can reduce back to the \VC{} instance~$(G', n'-\alpha)$.
As a result, most reduction rules for \independentset{} can be easily reformulated for \VC{}.
In this section we present slight reformulations of some stability transformations in the context of \VC{}.

The Struction Rule (short for STability number RedUCTION) has mainly been studied in the context of \independentset{}.
Intuitively, for a given vertex~$v$ the rule deletes~$N[v]$ and replaces it be a set of vertices corresponding to each non-edge in~$N(v)$.
If the set of newly created vertices is large, then~$k$ will be increased by ths rule.
For an example of the struction see \cref{fig:rr_struction}.

\begin{frule}[Struction \cite{EHW84}]
\label{rr:struction}
	Let~$v \in V$. Reduce~$(G,k)$ to~$(G', k')$ as follows:

	Let~$R = V \setminus N[v]$ and~$N(v) = \{ a_1, \dots, a_d \}$
	\begin{itemize}
	\item Remove~$N[v]$,
	\item add the following set of new vertices~$W = \{ v_{i,j} \mid 1 \leq i < j \leq d\ \text{ and } \{a_i, a_j\} \notin E\}$,
	\item for~$v_{i,j}, v_{k,l} \in W$ add the edge~$\{v_{i,j}, v_{k,l}\}$ if~$i \neq k$ or~$\{a_j, a_l\} \in E$, and
	\item for~$v_{i,j} \in W$ and~$u \in R$ add the edge~$\{v_{i,j}, u\}$ if~$\{a_i, u\} \in E$ or~$\{a_j, u\} \in E$.
	\end{itemize}

	Finally, set~$k' = k + |W| - d$.
\begin{figure}[t]
	\centering
	\begin{subfigure}[c]{.45\textwidth}
		\centering
		\includegraphics{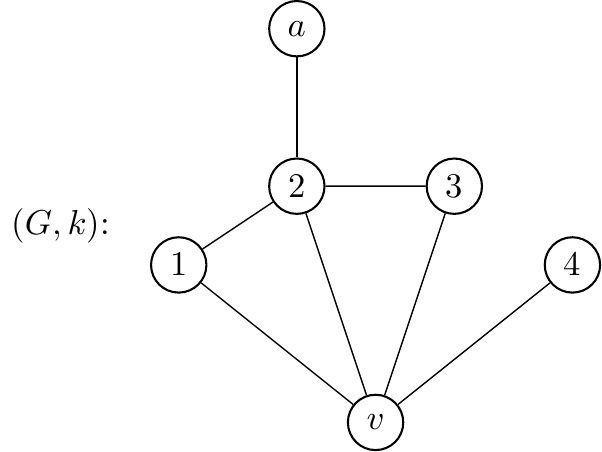}
	\end{subfigure}
	\begin{subfigure}[c]{.45\textwidth}
		\centering
		\includegraphics{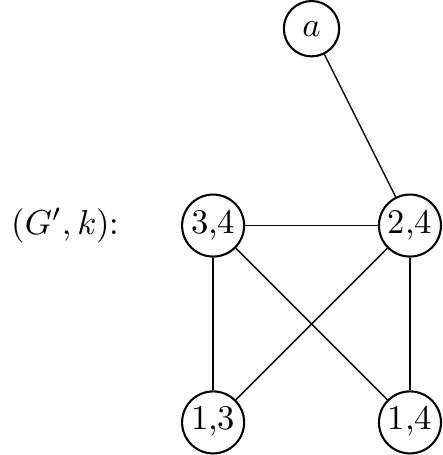}
	\end{subfigure}
\caption{An example for the application of \cref{rr:struction}. The vertex~$v$ has degree four and four missing edges in~$N(v)$.}
\label{fig:rr_struction}
\end{figure}
\end{frule}

The original proof of the safeness of \cref{rr:struction} by \citet{EHW84} relied on reformulations of pseudo-Boolean functions.
For a direct proof we refer to \citet{AHL03}.

The Struction Rule can always be applied to any vertex, runs in polynomial-time and reduces the stability number~$\alpha(G)$ by one~\cite{EHW84}.
Note that different permutations of~$N(v)$ can produce non-isomorphic~$G'$'s.

After at most~$n$ applications of the Struction Rule, an~$n$-vertex graph~$G$ can be reduced to an empty graph, as the stability number steadily decreases until it becomes zero~\cite{EHW84}.
Note that this does not immediately imply~$\text{P} = \text{NP}$, as a single application of the struction rule could increase the number of vertices to~$\Oh(n^2)$ and therefore for~$n$ applications exponential time and space might be needed.

In the following we would like to briefly analyze the cases in which the application of the struction rule leads to a decrease of~$k$.
Since~$|W|$ is equal to the number of missing edges in~$G[N(v)]$, $k$ is decreased by at least one if there are less than~$d$ missing edges in~$G[N(v)]$.
If there are~$d$ missing edges then~$k$ stays the same.
As a result, applying the struction to a vertex of degree at most three cannot increase~$k$.
However, if more than~$d$ edges are missing then~$k$ is increased by the Struction Rule.

The exhaustive application of the Struction Rule will always shrink the instance to an empty graph.
However, this may take exponential time in general.
To avoid this problem one may, for example, apply the Struction Rule only if it does not increase the number of vertices in the graph~\cite{GLS21}.
However, in this case there is no guarantee on getting an empty graph by exhaustive application of the data reduction rule.

An example of a rule which preserves the stability number is the Magnet Rule by \citet{HH91}.
For an example of the rule see \cref{fig:rr_magnet}.

\begin{frule}[Magnet \cite{HH91,HW09}]
\label{rr:magnet}
	Let~$a, b$ be two adjacent vertices and~$A = N(a) \setminus N[b], B = N(b) \setminus N[a], C = N(a) \cap N(b)$.
	If every vertex in~$A$ is adjacent to every vertex in~$B$, then reduce~$(G, k)$ to~$(G', k-1)$ as follows:
	\begin{itemize}
	\item remove~$a, b$,
	\item create a new vertex~$c$ and add all edges between~$c$ and~$C$.
	\end{itemize}
\end{frule}
\begin{figure}[t]
	\centering
	\begin{subfigure}[c]{.45\textwidth}
		\centering
		\includegraphics{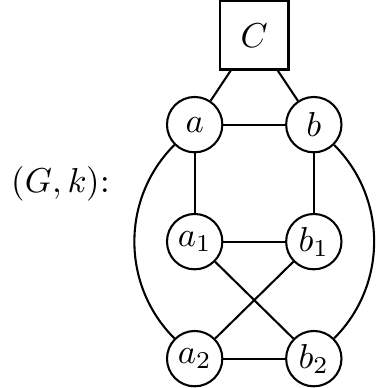}
	\end{subfigure}
	\begin{subfigure}[c]{.45\textwidth}
		\centering
		\includegraphics{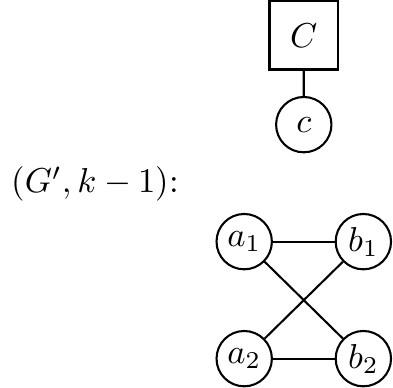}
	\end{subfigure}
\caption{An example for the application of \cref{rr:magnet}. The sets~$A = \{ a_1, a_2 \}, B = \{ b_1, b_2 \}$ do not necessarily have to be independent sets.}
\label{fig:rr_magnet}
\end{figure}

The Magnet Rule is a generalization of the Domination Rule (\cref{rr:domination}), which corresponds to the case where~$A = \emptyset$ or~$B = \emptyset$.
Similarly to the Struction rule, the Magnet Rule has been first discovered using pseudo-Boolean based arguments.

\subsubsection{Crown-based rules}

The crown rule by \citet{ACF04} looks for a so called \emph{crown} within the graph.
Once a crown is found, it can be shown that there always exists a minimum vertex cover which includes one part of the crown and excludes the rest.
The original crown definition is the following: 

\begin{definition}[\cite{ACF04}]
\label{rr:crown_def}
	A \emph{crown} is an ordered pair~$(H,I)$ of distinct vertex subsets from a graph~$G$ that satisfies the following criteria:
	\begin{itemize}
	\item~$H = N(I)$,
	\item~$I$ is a non-empty independent set, and
	\item the edges connecting~$H$ and~$I$ contain a matching in which all elements of~$H$ are matched.
	\end{itemize}
	The set~$H$ is also called the \emph{head} of the crown, and~$I$ the \emph{points}.
	\begin{figure}[t]
		\centering
		\includegraphics{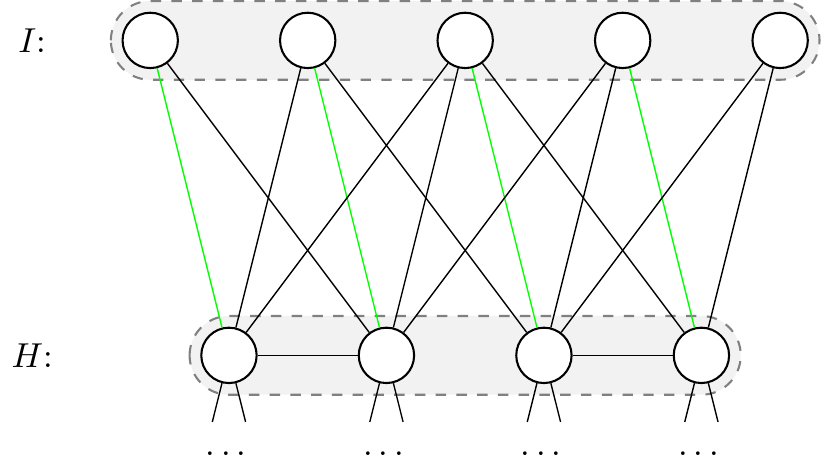}
		\caption{An example of a crown. A matching in which all vertices in~$H$ are matched is highlighted in green.}
		\label{fig:rr_crown}
	\end{figure}
\end{definition}

For an example of a crown see \cref{fig:rr_crown}.
The crown rule simply states that any crown can be directly removed from the graph~\cite{ACF04}.

\begin{frule}[Crown \cite{ACF04}]
\label{rr:crown}
Let~$(H,I)$ be a crown. Then delete~$H, I$ from~$G$ and decrease~$k$ by~$|H|$.
\end{frule}

Before the crown rule was discovered, a data reduction rule based on a linear programming relaxation of \VC{} was discovered by \citet{NT75}.
While at first glance they appear to be very different, they are in fact very closely related \cite{AI16,KKN21}.

\begin{frule}[LP \cite{NT75}]
\label{rr:lp}
	Let~$x \in \RR^n$ be an optimal solution to the following Linear Programming relaxation of \VC{}
	\begin{align*}
	& \text{minimize:}   & & \sum_{v \in V} x_v \\
	& \text{subject to:} & & x_u + x_v \geq 1 & & \text{ for all } \{u, v\} \in E \\
	&                    & & x_v \geq 0                                      & & \text{ for all } v \in V
	\end{align*}
	Let~$V_0 = \{v \in V | x_v < \frac12\}, V_{\frac12} = \{v \in V | x_v = \frac12\}, V_1 = \{v \in V | x_v > \frac12\}$. Delete~$V_0$ and~$V_1$ and decrease~$k$ by~$|V_1|$.
\end{frule}
\begin{lemma}[\cite{NT75}]
\cref{rr:lp} is safe.
\end{lemma}

\citet{NT75} show that there exists a minimum vertex cover of~$G$ which contains none of~$V_0$ and all of~$V_1$.
Furthermore they show there exists a half-integral solution to the LP relaxation, that is, there is an optimal solution~$x^*$ such that~$x_v^* \in \{0, \frac12, 1\}$ for all~$v$.
Additionally, they show that such a half-integral solution can be computed directly by solving a matching problem in a bipartite graph.
The total running time to find the solution is~$\Oh(m\sqrt{n})$ if the Hopcroft-Karp algorithm is used to compute the matching.

\citet{IOY14} furthermore show how to compute a half-integral solution with a minimum number of~$\frac12$ variables (an extreme half-integral solution) from any half-integral solution in linear time.
This can be achieved computing strongly connected components in a residual flow graph.
The well-established crown reduction by \citet{ACF04} is contained in the LP-reduction if an extreme half-integral solution is used~\cite{AI16}.
Furthermore, the graphs resulting from the exhaustive application of \cref{rr:crown} and \cref{rr:lp} are the same~\cite{KKN21}.
\cref{rr:lp} can be applied exhaustively in~$\Oh(m\sqrt{n})$ time~\cite{AI16}.

\citet{CKX10} provide a notion of a \emph{relaxed crown}, which becomes a normal crown if a vertex from its head is removed.
\begin{definition}
	A \emph{relaxed crown} in a crown-free graph~$G = (V,E)$ is a pair~$(H,I)$ such that
	\begin{itemize}
	\item~$I \subseteq V$ is an independent set in~$G$,
	\item~$H = \bigcup_{v \in I} N(v)$, and
	\item there exists a~$v \in H$, such that~$(H \setminus \{v\}, I)$ is a crown in~$G-v$
	\end{itemize}
\end{definition}
\begin{frule}[Relaxed crown \cite{CKX10}]
\label{rr:relcrown}
	Let~$G$ be crown-free and~$(H,I)$ be a relaxed crown.
	\begin{itemize}
	\item If~$H$ is an independent set, merge the vertices in~$I \cup H$ into a single vertex and decrease~$k$ by~$|I|$
	\item otherwise, remove~$I \cup H$ from~$G$ and decrease~$k$ by~$|H|$
	\end{itemize}
\end{frule}

\citet{AI16} use a Twin reduction rule, where~$u, v$ are twins if~$N(u) = N(v) = \{a, b, c\}$~\cite{XN13}.
The Twin Rule is a special case of \cref{rr:relcrown}.

\subsection{Relations between forward rules} \label{sect:rrule_relations}
In this section, we briefly summarize some relations between the forward rules in \cref{sect:forward_rules}.
One such relation is that of generalization.
We say a forward rule~$R_A$ \emph{generalizes} another rule~$R_B$ if~$R_B \subseteq R_A$.
This means, that rule~$R_B$ is merely a ``special case'' of rule~$R_A$.
\begin{figure}[t]
	\centering
	\includegraphics{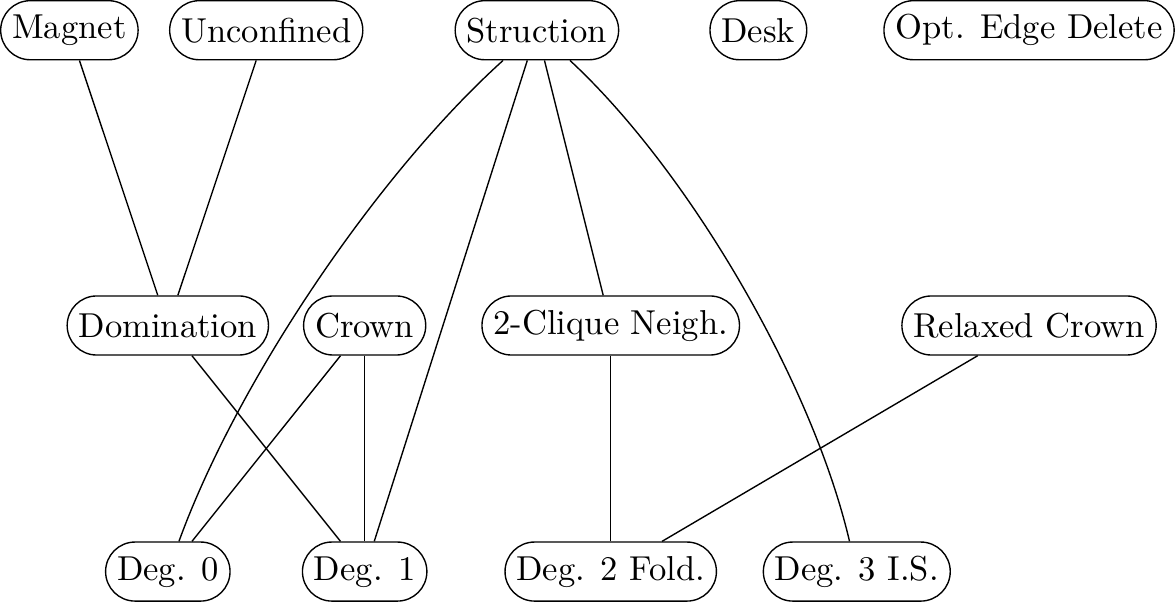}
	\caption{A diagram depicting rule generalization. An edge is drawn from a rule~$A$ to another rule~$B$ above it, if~$B$ generalizes~$A$. Note that formally the Domination Rule requires the addition of the Degree-0 Rule to generalize the Degree-1 Rule, and the Relaxed Crown is applicable only to crown-free graphs while Degree-2 Folding Rule has no such requirement, but we still draw these edges.}
	\label{fig:generalization}
\end{figure}

\cref{fig:generalization} summarizes which forward rules are special cases of the other rules.
We provide the 11 generalization proofs in \cref{appendix:proofs}, as they are rather repetitive.
Omitted from the figure is the LP Rule (\cref{rr:lp}), the exhaustive application of which is equivalent to the exhaustive application of the Crown Rule (\cref{rr:crown})~\cite{KKN21}.
Additionally, what the figure does not show is that the Magnet rule (\cref{rr:magnet}) can be realized as a combination of the Optional Edge Deletion Rule (\cref{rr:oe_delete}) and the Domination Rule (\cref{rr:domination})~\cite{AHL03}.

\subsection{Backward rules}\label{sect:backward_rules}
In this section, we formulate some backward rules based on the forward rules in \cref{sect:forward_rules}.
Recall that applying a backward rule corresponds to ``undoing'' a forward rule.
A forward rule~$R$ can be thought of as a transformation of an instance~$\mathcal{I}$ into a different instance~$\mathcal{I'}$, and the backward rule is the inverse transformation which takes~$\mathcal{I'}$ and produces~$\mathcal{I''}$, such that~$R$ applied on~$\mathcal{I''}$ produces~$\mathcal{I}$.
An important observation is that a data reduction rule may reduce multiple instances to the same instance~$\mathcal{I}$ or even one instance to different instances (due to nondeterminism of some rules).
As a result the inverse transformation is generally not unique.
This non-uniqueness in the backward rules allows for much more freedom in how the rule can be applied, compared to forward rules, as we will subsequently see.

To show that a rule is a backward rule of some corresponding forward rule, one has to show that by first applying the forward rule and then the backward rule to an instance~$\mathcal{I}$ one can always obtain the original instance~$\mathcal{I}$.

To prove the safeness of a backward rule, one only has to show that by applying the backward rule and then its corresponding forward rule one may always return to the original instance.
The safeness of the backward rule then follows from the safeness of the forward rule.

The illustrations we will see for each backward rule already hint that it is always possible to apply a backward rule and then a forward rule or vice versa to obtain an unchanged instance.
For brevity, most of these proofs will be omitted, as they are rather simple.

Due to the large number of forward rules described in \cref{sect:forward_rules}, not all corresponding backward rules will be formulated.
Instead, we primarily focus on simple backward rules.
Recall, that after applying a backward rule we aim to further apply forward or backward rules so that the graph is shrunk.
This is only possible with non-confluent rules.
However, it may be very challenging to prove or disprove confluence for each subset, or even each pair, of our forward rules.

We aim to develop backward rules that change the graph in a ``non-trivial'' way.
For example, we consider the addition of isolated vertices to the graph (undoing the Degree-0 Rule)  a ``trivial'' modification.
Likely, the insertion of isolated vertices will not aid in our final goal of shrinking the graph.

The first backward rule we want to introduce is that of Degree-2 Folding (\cref{rr:deg2}).
This backward rule is also called Vertex Splitting~\cite{Loz17}.
It allows one to split a vertex into three vertices, and distribute the neighbors of the original vertex among two of them.
This is at the cost of increasing the parameter~$k$ by one.
For an example see \cref{fig:undeg2}.

\begin{brule}[Backward Degree-2 Folding (Vertex Splitting)]
\label{rr:undeg2}
	Let~$v$ be a vertex. Split~$v$ as follows:
	\begin{itemize}
		\item create nonadjacent vertices~$a,b$,
		\item create edges adjacent to~$a$ or~$b$ such that~$N(a)\cup N(b) = N(v)$,
		\item delete any edges adjacent to~$v$, and
		\item create the edges~$\{v, a\}$ and~$\{v, b\}$.
	\end{itemize}
	Finally, increase~$k$ by one.
\begin{figure}[t]
	\centering
	\begin{subfigure}[c]{.45\textwidth}
		\centering
		\includegraphics{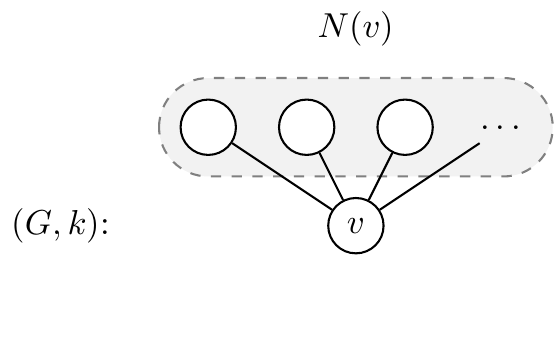}
	\end{subfigure}
	\begin{subfigure}[c]{.45\textwidth}
		\centering
		\includegraphics{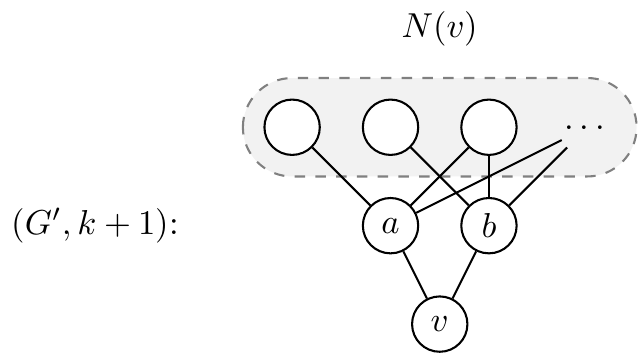}
	\end{subfigure}
\caption{An illustration of the Backward Degree-2 Folding Rule (\cref{rr:undeg2}). This operation described by the rule is sometimes also called Vertex Splitting.}
\label{fig:undeg2}
\end{figure}
\end{brule}
\begin{lemma}
\cref{rr:undeg2} is safe.
\end{lemma}
\begin{proof}
Let~$(G,k)$ be an instance of \VC{}, $v \in V$ any vertex and~$(G',k+1)$ the instance obtained from splitting the vertex~$v$ using \cref{rr:undeg2} (in any way). 
Let~$v,a$, and~$b$ be the vertices in~$G'$ into which~$v$ was split.
The vertex~$v$ in~$G'$ has a degree of two, and its two neighbors~$a$ and~$b$ are nonadjacent.
Consequently, the Degree-2 Folding Rule can be applied to~$v$ in~$G'$.
The rule merges the vertices~$v,a$, and~$b$ into a single vertex~$v'$, such that the neighborhood of~$v'$ is~$N_G'(\{v,a,b\})$, and decreases~$k$ by one.
This results in the instance~$(G'', k$).
Notice that~$N_{G''}(v') = N_G(v)$.
As a consequence, $G$ and~$G''$ are isomorphic.

Because the Degree-2 Folding Rule is safe, the instances~$(G'', k)$ and~$(G',k+1)$ are equivalent, which implies that~$(G,k)$ and~$(G',k+1)$ are also equivalent.
\end{proof}

Vertex Splitting as described by \cref{rr:undeg2} can be used to show that \VC{} is NP-hard even on graphs with maximum degree at most three \cite{PV97,Loz17}.
This statement can be easily proven as follows.
Given a graph~$G$, split each vertex~$v$ with~$\deg_G(v) = d > 3$ into the vertices~$v,a,b$ using \cref{rr:undeg2} by assigning two of~$v$'s neighbors to~$a$, and the remaining to~$b$.
This way~$\deg(v) = 2, \deg(a) = 3, \deg(b) = d - 1$, and the degrees of the remaining vertices are unchanged.
This is repeated as long as as long as there exists a vertex of degree at least four.

Another interesting observation can be made.
If a vertex~$v$ has two nonadjacent neighbors~$x,y$ (which is always true if the Domination Rule (\cref{rr:domination}) is not applicable), then we can trigger the Degree-3 Independent Set Rule as follows (see also \cref{fig:rr_undeg2_deg3}:
split~$v$ into three vertices~$v,a,b$ as described in \cref{rr:undeg2} and assign~$x,y$ to~$a$ and the remaining neighbors of~$v$ to~$b$.
As a result, the vertex~$a$ has degree-3 and its neighborhood is an independent set, which means the Degree-3 Independent Set Rule can be applied to~$a$.
After applying these forward and backward rules~$k$ is increase by one.
However, the resulting neighborhood of the vertex~$v$ has been ``rewired'' substantially.
Therefore, cases may exist where other rules may now be applied to further shrink the graph.

\begin{figure}[t]
	\centering
	\begin{subfigure}[c]{.45\textwidth}
		\centering
		\includegraphics{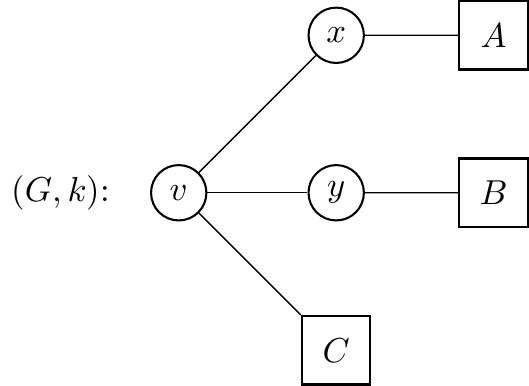}
		\caption{The original graph~$G$}
	\end{subfigure}
	\begin{subfigure}[c]{.45\textwidth}
		\centering
		\includegraphics{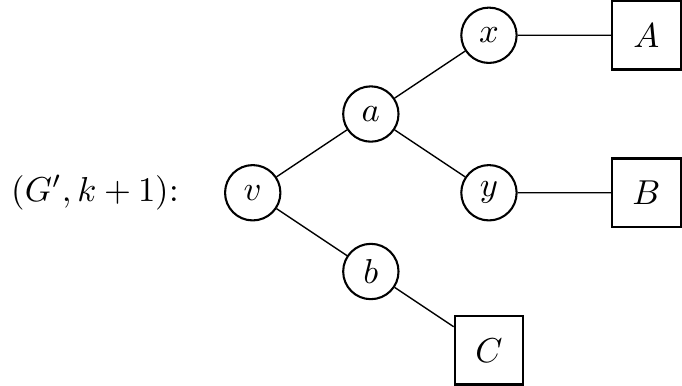}
		\caption{The graph obtained by applying the backward degree 2 folding rule to~$v$ in~$G$.}
	\end{subfigure}
	\begin{subfigure}[c]{.45\textwidth}
		\centering
		\includegraphics{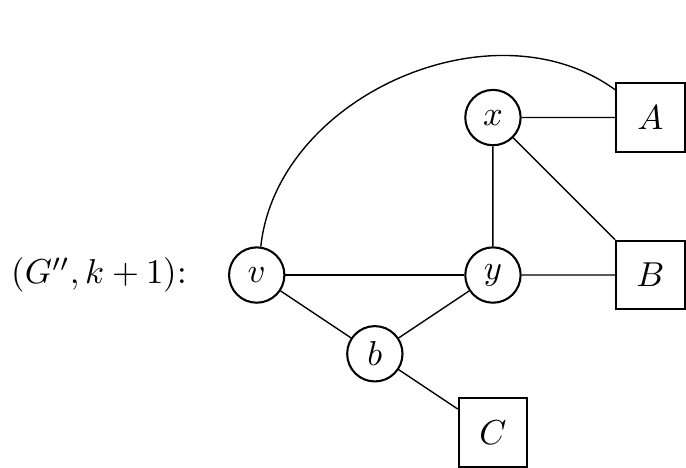}
		\caption{The graph obtained by applying the degree 3 independent set rule to~$a$ in~$G'$.}
	\end{subfigure}
\caption{An example for how the degree 3 independent set rule can be triggered using the backward degree 2 folding (vertex splitting).} 
\label{fig:rr_undeg2_deg3}
\end{figure}

We next provide a backward rule for the Degree-3 Independent Set forward rule.
See also \cref{fig:rr_undeg3} for an example.
\begin{brule}[Backward Degree-3 Independent Set]
\label{rr:undeg3}
	Let~$abc$ be a an induced~$P_3$ and~$S = \{a, b, c\}$. If for all~$u \in N(S)$ we have~$|N(u) \cap S| \geq 2$, then
	\begin{itemize}
		\item delete the edges~$\{a,b\}, \{b,c\}$,
		\item for all~$u \in N(S)$:
		\begin{itemize}
			\item if~$S \subseteq N(u)$ then optionally delete one of the edges connecting~$u$ to~$S$,
			\item if~$a \notin N(u)$, then delete~$u$'s edge to~$b$,
			\item if~$b \notin N(u)$, then delete~$u$'s edge to~$c$,
			\item if~$c \notin N(u)$, then delete~$u$'s edge to~$a$.
		\end{itemize}
		\item create a new vertex~$v$ and together with all edges from~$v$ to~$a,b$, and~$c$.
	\end{itemize}
\begin{figure}[t]
	\centering
	\begin{subfigure}[c]{.45\textwidth}
		\centering
		\includegraphics{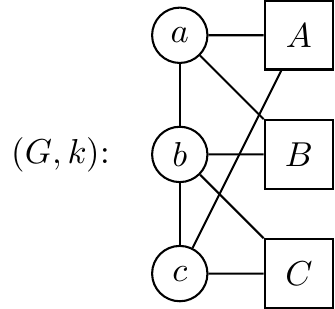}
	\end{subfigure}
	\begin{subfigure}[c]{.45\textwidth}
		\centering
		\includegraphics{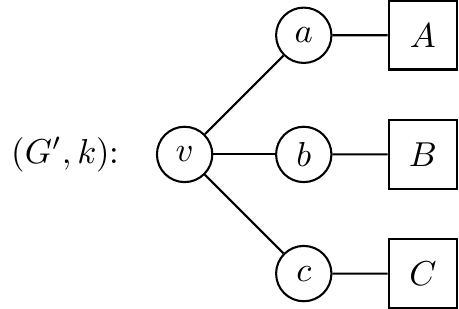}
	\end{subfigure}
\caption{An example for the Backward Degree-3 Independent Set Rule (\cref{rr:undeg3}). The case where a vertex is adjacent to~$a,b$ and~$c$ in~$G$ is not depicted.}
\label{fig:rr_undeg3}
\end{figure}
\end{brule}
The example in \cref{fig:rr_undeg2_deg3} also demonstrates that one can sometimes undo the Degree-3 Rule and then apply the Degree-2 Folding rule to shrink the graph.

Additionally, we provide a special case of the backward rule for the 2-Clique Neighborhood Rule (\cref{rr:cn}).
In this simplified rule we only look for a very small clique, namely an edge.
For an illustration see \cref{fig:uncn}.
\begin{brule}[Backward 2-Clique Neighborhood (special case)]
\label{rr:uncn}
	Let~$a,b$ be two adjacent vertices and~$C = N(a) \cap N(b)$. Then
	\begin{itemize}
		\item create two new vertices~$v,c$,
		\item connect~$v$ to~$a,b,c$,
		\item delete~$a$ and~$b$'s edges to~$C$, and
		\item create all edges between~$c$ and~$C$.
	\end{itemize}
	Finally, increase~$k$ by one.
\begin{figure}[t]
	\centering
	\begin{subfigure}[c]{.45\textwidth}
		\centering
		\includegraphics{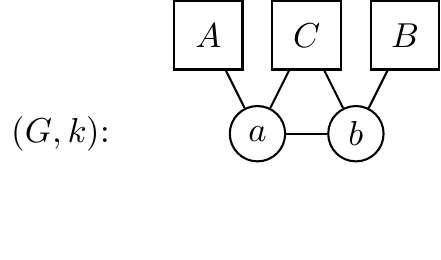}
	\end{subfigure}
	\begin{subfigure}[c]{.45\textwidth}
		\centering
		\includegraphics{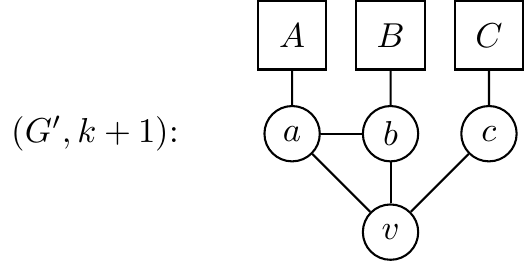}
	\end{subfigure}
	\caption{An illustration of \cref{rr:uncn}.}
	\label{fig:uncn}
\end{figure}
\end{brule}
\begin{lemma}
\cref{rr:uncn} is safe.
\end{lemma}
The safeness of \cref{rr:uncn} follows from the safeness of \cref{rr:cn}.

We also consider two backward rules which only create a new vertex and some edges adjacent to it.
The first rule of this type is the Backward Domination.
\begin{brule}[Backward Domination]
\label{rr:undom}
	Let~$v$ be any vertex and~$S \subseteq V$.
	Then create a new vertex~$u$ and add all edges between~$u$ and~$S \cup N[v]$, and increase~$k$ by one.
\begin{figure}[t]
	\centering
	\begin{subfigure}[c]{.45\textwidth}
		\centering
		\includegraphics{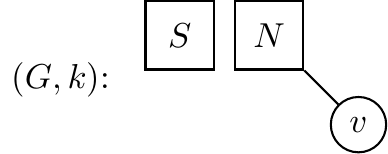}
	\end{subfigure}
	\begin{subfigure}[c]{.45\textwidth}
		\centering
		\includegraphics{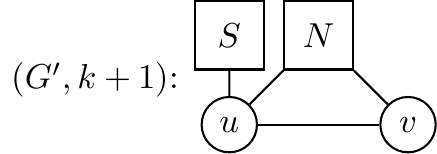}
	\end{subfigure}
\caption{An illustration for \cref{rr:undom}. The set~$N$ represents the neighbors of~$v$. The vertex~$u$ is added in such a way that it dominates~$v$ in~$G'$. The set~$S$ can be chosen arbitrarily.}
\label{fig:rr_undom}
\end{figure}
\end{brule}
\begin{lemma}
\cref{rr:undom} is safe.
\end{lemma}
The safeness of \cref{rr:undom} follows from the safeness of \cref{rr:domination}.

The second backward rule which simply only creates a new vertex and some edges incident to it is the Backward Unconfined Rule.
\begin{brule}[Backward Unconfined]
\label{rr:ununconf}
	Let~$S \subseteq V$ and~$G'$ be the graph obtained from~$G$ by adding a new vertex~$v$ whose neighborhood is~$S$.
	If~$v$ is an unconfined vertex (\cref{rr:unconfined}) then change~$G$ for~$G'$ and increase~$k$ by one.
\end{brule}
\begin{lemma}
\cref{rr:ununconf} is safe.
\end{lemma}
The safeness of \cref{rr:ununconf} follows from the safeness of \cref{rr:unconfined}.

The last backward rule which we will consider is that of the Optional Edge Deletion forward rule.
The edge between two nonadjacent vertices~$a$ and~$b$ may be inserted if there is third vertex~$c$ adjacent to either~$a$ or~$b$ and~$N(c) \subseteq N(a) \cup N(B)$.
The two vertices~$a$ and~$b$ can be thought of as jointly dominating the vertex~$c$.
\begin{brule}[Backward Optional Edge Deletion / Optional Edge Insertion]
\label{rr:oe_insert}
Let~$a,b,c \in V$ be pairwise distinct vertices such that~$\{a, b\} \notin E$, the vertex~$c$ is adjacent to exactly one of~$a$ or~$b$ and~$N(c) \subseteq N(a) \cup N(b)$.
Then insert the edge~$\{a, b\}$.
\begin{figure}[t]
	\centering
	\begin{subfigure}[c]{.45\textwidth}
		\centering
		\includegraphics{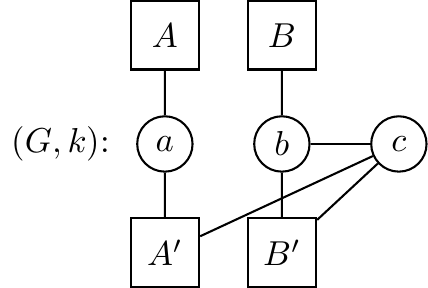}
	\end{subfigure}
	\begin{subfigure}[c]{.45\textwidth}
		\centering
		\includegraphics{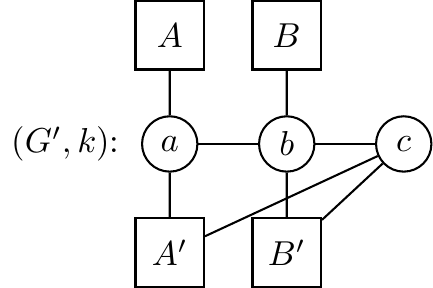}
	\end{subfigure}
\caption{An example for the application of \cref{rr:oe_insert}.}
\label{fig:rr_oe_insert}
\end{figure}
\end{brule}
\begin{lemma}
\cref{rr:oe_insert} is safe.
\end{lemma}
Again, the safeness of \cref{rr:oe_insert} follows from the safeness of \cref{rr:oe_delete}.

\section{Generalization proofs}\label{appendix:proofs}
\begin{figure}[t!]
	\centering
	\includegraphics{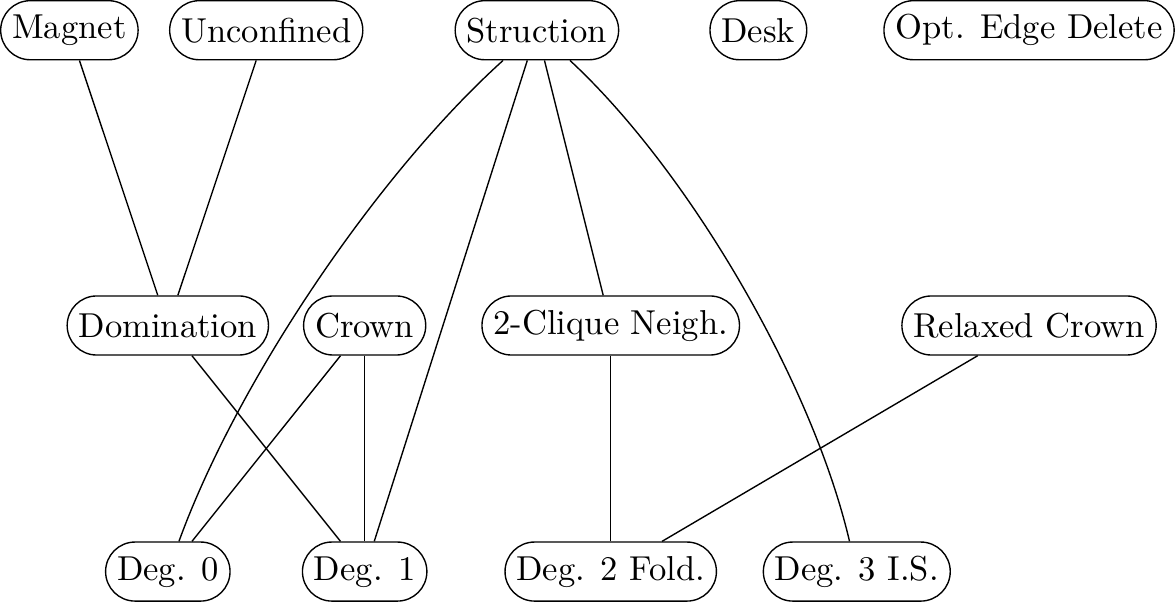}
	\caption{A diagram depicting rule generalization. An edge is drawn from a rule~$A$ to another rule~$B$ above it, if~$B$ generalizes~$A$. Note that formally the Domination Rule requires the addition of the Degree-0 Rule to generalize the Degree-1 Rule, and the Relaxed Crown only generalizes the Degree-2 Folding Rule when the graph is crown-free, but we still draw these edges.}
	\label{fig:generalization2}
\end{figure}
In this appendix we give the missing proofs for hierarchy of generalization among forward rules, which was presented in \cref{sect:rrule_relations}.
For convenience, we include a copy of the figure here as \cref{fig:generalization2}.

\begin{lemma}
The Magnet Rule (\Cref{rr:magnet}) is a generalization of the Domination Rule (\Cref{rr:domination}).
\end{lemma}
\begin{proof}
Let~$v$ be a vertex which dominates another vertex~$u$, i.\,e., $N[u] \subseteq N[v]$. 
For~$A = N(u) \setminus N[v] = \emptyset{}$, $B = N(v) \setminus N[u]$, and~$C = N(u) \cap N(v)$ we see that all vertices in~$A$ are adjacent to~$B$, and so is~$u$ and~$v$.
It follows that the Magnet rule can be applied to~$(u,v)$.
As a result of applying the rule, $u$ and~$v$ are deleted and a vertex~$c$ is created and connected to the vertices in~$C = N(u) \setminus \{v\}$.
The vertex~$c$ corresponds to~$u$ in the graph obtained by deleting~$v$ as part of applying the Domination Rule.
Therefore, applying the Domination Rule to~$u,v$ results in the same instance as applying the Magnet Rule to~$u,v$.
The parameter~$k$ is decreased by one in both rules.
\end{proof}

\begin{lemma}
The Unconfined Rule (\Cref{rr:unconfined}) is a generalization of the Domination Rule (\Cref{rr:domination}).
\end{lemma}
\begin{proof}
Let~$v$ be a vertex which dominates another vertex~$u$, i.\,e., $N[u] \subseteq N[v]$. 
Further let~$S = \{v\}$.
Notice that~$u \in N(S)$, $|N(u) \cap S| = 1$, and~$|N(u) \setminus N[S]| = 0$.
Therefore after one iteration of the procedure in the Unconfined Rule yes will be returned.
That means~$v$ is unconfined, and will be removed by both the Domination and Unconfined rules.
The parameter~$k$ is decreased by one by both rules.
\end{proof}

\begin{lemma}
The Struction Rule (\Cref{rr:struction}) is a generalization of the Degree-0 Rule (\Cref{rr:deg0}).
\end{lemma}
\begin{proof}
Let~$v$ be a degree-0 vertex.
The Degree-0 Rule simply deletes~$v$.
The Struction Rule deletes~$N[v]$ and adds a set~$W$ of new vertices for each missing edge in~$N(v)$, which in this case is empty.
Therefore the Struction just deletes~$v$.
The parameter~$k$ is not changed in either case.
\end{proof}

\begin{lemma}
The Struction Rule (\Cref{rr:struction}) is a generalization of the Degree-1 Rule (\Cref{rr:deg1}).
\end{lemma}
\begin{proof}
Let~$v$ be a degree-1 vertex and~$u$ its unique neighbor..
The Degree-1 Rule simply deletes~$u$ and~$v$.
The Struction Rule deletes~$N[v]$ and adds a set~$W$ of new vertices for each missing edge in~$N(v)$, which in this case is empty.
Therefore the Struction also just deletes~$u$ and~$v$.
The parameter~$k$ is decreased by one by both rules.
\end{proof}

\begin{lemma}
The Struction Rule (\Cref{rr:struction}) is a generalization of the 2-Clique Neighborhood Rule (\Cref{rr:cn}).
\end{lemma}
\begin{proof}
Let~$v$ be a vertex and~$N(v) = \{ a_1, \dots, a_d \}$, such that~$C_1 = \{a_1, \dots, a_{p}$ and~$C_2 = \{a_{p+1}, \dots, a_d\}$ form a partition of~$N(v)$ that satisfies the condition of the 2-Clique Neighborhood Rule.
That means, $p = |C_1| \geq |C_2| = d-p$, $C_1$ and~$C_2$ are both cliques and each vertex in in~$C_1$ is missing exactly one edge to a vertex in~$C_2$.
For~$a_i \in C_1$ let ~$m(a_i)$ be the vertex in~$C_2$ not adjacent to~$a_i$.

The 2-Clique Neighborhood Rule deletes~$v$ and~$C_2$ and connects~$a_i$ to all vertices in~$N_G(m(a_i)) \setminus N[v]$.

The Struction Rule deletes~$N[v]$ and instead creates a set~$W = \{v_{i,j} \mid a_i \in C_1, a_j = m(a_i)\}$ of new vertices, which become a clique.
Furthermore~$v_{i,j}$ is connected to~$N_G(\{a_i, a_j\}) \setminus N[v]$.
This means the vertex~$v_{i,j}$ corresponds to the vertex~$a_i$ in the case where the 2-Clique Neighborhood was applied.

The parameter~$k$ is decreased by~$d-p$ by both rules.
\end{proof}

\begin{lemma}
The Struction Rule (\Cref{rr:struction}) is a generalization of the Degree-3 Independent Set Rule (\Cref{rr:deg3}).
\end{lemma}
\begin{proof}
Let~$v$ be a degree-3 vertex and~$N(v) = \{ a_1, a_2, a_3 \}$ an independent set.

The Degree-3 Independent Set Rule adds the edges~$\{a_1, a_2\}$, $\{a_2, a_3\}$, and all edges between~$a_1$ and~$N(a_2)$, $a_2$ and~$N(a_3)$ and between~$a_3$ and~$N(a_1)$.
It also deletes~$v$.

The Struction Rule instead deletes~$N[v]$ and creates a set~$W = \{v_{1,2}, v_{1,3}, v_{2,3}\}$ and adds the edges~$\{v_{1,2}, v_{2,3}\}, \{v_{1,3}, v_{2,3}\}$.
Furthermore~$v_{i,j}$ is connected to~$N_G(\{a_i, a_j\}) \setminus N[v]$.
This means the vertices~$v_{1,2}, v_{2,3}$ and~$v_{1,3}$ correspond to the vertices~$a_1, a_2$ and~$a_3$, respectively, in the case where the Degree-3 Rule was applied.

The parameter~$k$ is not changed by either rule.
\end{proof}

\begin{lemma}
The Domination Rule (\Cref{rr:domination}) in combination with the Degree-0 Rule (\Cref{rr:deg0}) is a generalization of the Degree-1 Rule (\Cref{rr:deg1}).
\end{lemma}
\begin{proof}
Let~$v$ be a degree-1 vertex and~$u$ its unique neighbor.
The vertex~$u$ dominates~$v$, and by the Domination rule can be removed.
The vertex~$v$ then becomes a degree-0 vertex and it may be deleted using the Degree-0 Rule.
The Degree-1 Rule simply directly deletes~$u$ and~$v$.
Both rules decrease~$k$ by one.
\end{proof}

\begin{lemma}
The Crown Rule (\Cref{rr:crown}) is a generalization of the Degree-0 Rule (\Cref{rr:deg0}).
\end{lemma}
\begin{proof}
Let~$v$ be a degree-0 vertex.
It is easy to see that~$(\emptyset, \{v\})$ forms a crown.
Both rules simply remove the vertex~$v$ without changing~$k$.
\end{proof}

\begin{lemma}
The Crown Rule (\Cref{rr:crown}) is a generalization of the Degree-1 Rule (\Cref{rr:deg1}).
\end{lemma}
\begin{proof}
Let~$v$ be a degree-1 vertex and~$u$ its unique neighbor.
It is easy to see that~$(\{u\}, \{v\})$ forms a crown.
Both the Crown and Degree-1 rules remove the vertices~$u$ and~$v$ and decrease~$k$ by one.
\end{proof}

\begin{lemma}
The 2-Clique Neighborhood Rule (\Cref{rr:cn}) is a generalization of the Degree-2 Folding Rule (\Cref{rr:deg2}).
\end{lemma}
\begin{proof}
Let~$v$ be a degree-2 vertex and~$a, b$ its neighbors, which are non-adjacent.
The Degree-2 Folding Rule deletes~$v, a$ and~$b$ and creates a new vertex~$c$ whose neighborhood is~$(N(a) \cup N(b)) \setminus \{v\}$.

The 2-Clique Neighborhood Rule deletes~$v$ and~$b$ (or~$a$), and connects~$a$ with~$N(b) \setminus \{v\}$.
Therefore the vertex~$a$ corresponds to the vertex~$c$.
Both the rules decrease~$k$ by one.
\end{proof}

\begin{lemma}
The Relaxed Crown Rule (\Cref{rr:relcrown}) is a generalization of the Degree-2 Folding Rule (\Cref{rr:deg2}) in crown-free graphs.
\end{lemma}
\begin{proof}
Let~$G$ be crown-free and~$v$ be a degree-2 vertex and~$a, b$ its neighbors, which are non-adjacent.
The Degree-2 Folding Rule merges~$v, a$ and~$b$ into a single vertex.

The pair~$(\{a,b\}, \{v\})$ is a relaxed crown, where the head is an independent set.
The Relaxed Crown Rule therefore merges~$v, a$ and~$b$ into a single vertex.
Both the rules decrease~$k$ by one.
\end{proof}

\end{document}